\documentclass[acmsmall,screen,nonacm]{acmart}

\usepackage{booktabs}
\usepackage{flafter}
\usepackage{mathtools}
\usepackage{mathpartir}
\usepackage{array}
\usepackage{enumitem}
\usepackage{listings}
\usepackage{tikz}
\usetikzlibrary{arrows.meta,positioning,fit,calc}
\usepackage{xspace}
\usepackage[capitalise,nameinlink]{cleveref}
\crefname{lstlisting}{Listing}{Listings}
\Crefname{lstlisting}{Listing}{Listings}

\newcommand{\sys}{\textsf{numoracle}\xspace}
\newcommand{\Ehp}{E_{\mathrm{hp}}}
\newcommand{\Wt}{W_{t}}
\newcommand{\fl}{\mathrm{fl}}

\newcommand{\LIB}{\mathrm{LIB}}
\newcommand{\orig}{\mathrm{orig}}
\newcommand{\dorig}{\mathrm{dorig}}
\newcommand{\ulp}{\mathrm{ulp}}
\newcommand{\Bimpl}{B_{\mathrm{impl}}}
\newcommand{\Rp}{\mathcal{R}}
\newcommand{\Cls}{\mathcal{K}(\Rp,\Pi)}
\newcommand{\code}[1]{\texttt{#1}}

\newcommand{\para}[1]{\paragraph{#1}}

\lstdefinestyle{py}{language=Python, basicstyle=\ttfamily\footnotesize, columns=fullflexible,
  keepspaces=true, commentstyle=\color{black!60}, keywordstyle=\bfseries,
  showstringspaces=false, xleftmargin=1em, frame=none}

\setcopyright{none}
\acmDOI{}
\acmJournal{PACMPL}
\acmVolume{1}
\acmNumber{OOPSLA}
\acmYear{2027}

\begin{document}

\title{Beyond atol and rtol: Sound Test Oracles for Low-Precision GPU Kernels}

\author{Xinyi Shi}
\affiliation{\institution{Duke University}\city{Durham}\state{North Carolina}\country{USA}}

\begin{abstract}
Low-precision GPU kernels are commonly tested against a reference using hand-picked absolute and relative tolerances.
Legitimate numerical error varies with the input, format and precision plan, so a fixed tolerance trades false alarms
for missed faults. We present \sys, which derives per-element acceptance regions from precision-annotated reference
programs. Its calculus tracks shared rounding errors, provenance-ordered products, block-scaled quantizers and
calibrated math-library errors. A rounding-interval check detects faulty stores that symmetric bounds can admit.
Soundness is conditional on explicit rounding, hardware, domain and implementation assumptions; a fixed-input mode
removes the zero-mean rounding assumption.

We evaluate 13 Triton operators from Liger-Kernel, Unsloth and vLLM on two GPU generations. The main campaigns
produce no false alarm on 33,960 correct-kernel and behavior-preserving-variant runs per GPU. Under input shift,
the GPU campaigns record no false alarm while tuned tolerances flag up to 30\% of correct runs; final-version
reference checks identify inputs without complete finite-range certificates. On the RTX PRO 6000, \sys kills 116 of 124 effective mutants, versus 59--83
for four fixed baselines. Among 400 tuned tolerance configurations, requiring at most 1\% false alarms in every
evaluated Liger input group permits at most 31 of 53 main-mutant kills, versus \sys's 50. Historical replays expose
accuracy faults in eight change pairs without rejecting their matched corrected controls; a ninth change concerns
rounding parity.

\end{abstract}



\maketitle

\section{Introduction}\label{sec:intro}

Hand-written GPU kernels implement normalization, activation, attention and loss operations in libraries such as
Liger-Kernel~\cite{hsu2024liger}, vLLM~\cite{kwon2023vllm} and Unsloth. Written in Triton~\cite{tillet2019triton}
or CUDA, they combine low-precision storage with higher-precision intermediates and approximate math intrinsics.

Reduction order and rounding positions change a correct kernel's output. Tests therefore accept differences up to $\mathit{atol}+\mathit{rtol}|\mathit{ref}|$, usually
against a PyTorch reference~\cite{barr2015oracle}. These hand-picked tolerances can both hide faults and reject
correct outputs. A vLLM Mamba2 test used $\mathit{atol}=\mathit{rtol}=0.5$ before its memory-offset and test-coverage
fix~\cite{vllm21783}; Liger's bf16 SwiGLU MLP test uses $\mathit{atol}=10^4$~\cite{ligertests}. Conversely,
PyTorch's defaults flag 24\% of correct Liger runs in our study, often because the low-precision reference is less
accurate than the kernel (\cref{sec:three-failures}).

Allowable error depends on computation and input: a sum's error scales with term magnitudes despite cancellation,
and normalization can amplify errors when row statistics are small. Both depend on the \emph{precision plan}:
intermediate rounding positions, formats and block scales. A tolerance covering the hardest inputs can sacrifice
sensitivity elsewhere. Tuned baselines expose this trade-off across input families, shapes and row statistics
(\cref{sec:rq3}).

A reference program and precision plan specify an implementation class, fixing numerical obligations but leaving
reduction order open. The analysis derives its input-specific acceptance region. Under an explicit zero-mean error
model, probabilistic rounding analysis tightens linear-algebra bounds from worst-case $nu$ to $\sqrt{n}\,u$
scaling~\cite{higham2019new}. A whole-program oracle must compose bounds through tensor operations, track rounding
events shared by broadcasting and reuse, model quantizers and non-RN hardware accumulators~\cite{fasi2021tensor},
and account for the final output grid.

We present \sys, an executable analysis of precision-annotated tensor references. Shadow tensors propagate float64
reference values, separating deterministically bounded residuals from modeled random errors; axis keys and event
provenance preserve shared-error dependence across tensor operations. Per-element error-bound and rounding-interval
checks judge kernel outputs. The latter can reject incorrect stores admitted by the symmetric bound, adding
detections on individual inputs but no additional killed mutant in the main campaign. The per-call false-alarm claim assumes the
rounding, hardware, domain and implementation contracts in \cref{sec:assumptions}; probability is over model
executions, not repeated deterministic GPU runs. Fixed-input mode removes the zero-mean assumption but retains the
other numerical contracts.

We evaluate 13 Triton operators in bf16, fp16 and fp32 on an RTX PRO 6000 Blackwell and an A100. Correct-kernel and
variant campaigns had no false alarms; under input shift, finite-range checks identify uncertifiable inputs
(\cref{sec:rq1}). Reruns confirmed LayerNorm and SSD checks with the final implementation. On the RTX PRO 6000, \sys killed
116/124 effective mutants, versus 59--83 for fixed oracles. Historical-bug replays test whether fixes eliminate
rejection; CPU emulations and targeted witnesses separately evaluate deeper compositions and dependence rules
(\cref{sec:evaluation}). Our contributions are:
\begin{itemize}
\item \emph{Problem.} We specify kernel correctness by a class of implementations and derive input-specific acceptance regions, exhibiting inputs on which fixed tolerances lose sensitivity or raise false alarms (\cref{sec:problem,sec:rq3}).
\item \emph{Calculus.} We give a dependence-aware probabilistic error calculus for whole low-precision tensor
  programs, with shared random components, provenance-ordered products, block-scaled quantizers and calibrated
  math-library profiles, and prove it sound under a stated rounding model (\cref{sec:calculus,sec:soundness}).
\item \emph{Checks and limits.} We introduce a rounding-interval check that detects sub-ulp faults, and
  state a single-run indistinguishability limit for oracles that accept every legal output (\cref{sec:single-run}).
\item \emph{Tool and study.} We implement \sys and evaluate it on 13 Triton operators from three libraries on two
  GPUs, against 400 tuned tolerance configurations and with nine historical bugs
  (\cref{sec:implementation,sec:evaluation}). The library, references, mutants and raw data are in the artifact.
\end{itemize}

Prior work gives probabilistic GPU matrix-product checksum thresholds~\cite{braun2014aabft,gao2026vabft}, calibrates
tolerances from correct runs~\cite{sarkar2026calibration} and measures benchmark-checker blind spots~\cite{du2026checker}.
Our contribution composes acceptance regions, shared-error tracking and a call-level probability budget through
tensor references (\cref{sec:related}). Through broadcasting, products and nonlinear operations, our rules preserve
the probability bound's conditions: shared events remain correlated, event order justifies tighter products, and
residuals are conservatively bounded.
This composition gives a call-level guarantee and converts pre-store bounds into output-grid intervals.
\section{Background and Motivation}\label{sec:background}

\subsection{Formats and Precision Plans}\label{sec:formats}

A binary floating-point format with $p$ significand bits rounds a real $z$ to a representable value
$\fl(z) = z(1+\varepsilon)$ with $|\varepsilon| \le u = 2^{-p}$ under round-to-nearest (RN), outside the subnormal
range. Current kernels use fp32 ($u = 2^{-24}$) for accumulators and row statistics, fp16 ($2^{-11}$) and bf16
($2^{-8}$) for activations and weights, FP8 E4M3 ($2^{-4}$, largest value 448) and E5M2 ($2^{-3}$) for GEMM inputs
and gradients, and block-scaled 4-bit formats. These store each element in E2M1 ($2^{-2}$), whose magnitudes are
$\{0, 0.5, 1, 1.5, 2, 3, 4, 6\}$, and share one scale per block: MXFP4 uses 32-element blocks with a power-of-two
scale~\cite{rouhani2023mx}, and NVFP4 uses 16-element blocks with an E4M3 scale and an fp32 tensor
scale~\cite{nvidia2025nvfp4}.

Precision plans specify rounding positions. Typical Triton normalization kernels load bf16 inputs, convert and
reduce in fp32, and store in bf16; matrix kernels accumulate products in fp32 on tensor cores. Plans intentionally
differ: Liger's RMSNorm offers bf16 or fp32 weight application. Kernels also use approximate math intrinsics: Triton
lowers fp32 \code{exp} to \code{ex2.approx} applied to $x\log_2 e$, and fp32 division to \code{div.full}, whose
specified errors exceed correct rounding's half ulp~\cite{tritonlowering,nvidiaptx}.

\subsection{Current Kernel-Testing Practice}\label{sec:testing-today}

Kernel tests usually compare output $o$ with the same-dtype PyTorch reference $r$, requiring
$|o - r| \le \mathit{atol} + \mathit{rtol}\cdot|r|$ elementwise. PyTorch's \code{assert\_close} defaults are
$\mathit{rtol} = 1.6\times10^{-2}$ for bf16, $10^{-3}$ for fp16 and $1.3\times10^{-6}$ for fp32, each with
$\mathit{atol} = 10^{-5}$~\cite{pytorchtesting}. Libraries override these per test. AMD's aiter RFC proposes replacing
hard-coded $\mathit{atol} = 100$ in mixture-of-experts tests with format-derived tolerances:
$\mathit{rtol} = f\cdot\max(u_c, u_o, K u_{\mathit{acc}})$ and
$\mathit{atol} = f\cdot 2^{\lfloor\log_2\max|r|\rfloor}\cdot\max(u_c, 2u_o, K u_{\mathit{acc}})$, for compute,
output and accumulation precisions, reduction length $K$ and safety factor $f$~\cite{aiter5031}.

Tolerances drift toward looseness. We mined the full histories of 11 open-source kernel libraries (Liger-Kernel,
Unsloth, vLLM, FlashAttention, xFormers, TransformerEngine, FlashInfer, SGLang, torchao, bitsandbytes and apex, up to
1 October 2026) for test-code commits changing $\mathit{atol}$, $\mathit{rtol}$ or cosine-similarity thresholds.
The detector reports 247 candidates: 125 with loosenings, 70 with tightenings (13 both), and 65 classified only as
literal-to-formula changes. A random sample of 40 reports per direction (78 distinct commits) was checked against
the diffs (labeling and review: \cref{sec:setup}); 35 loosening and 25 tightening reports genuinely changed allowances.
Adjusting detector counts by these estimated precisions gives a loosening-to-tightening ratio of about 2.5 to 1.
At repository heads, 178 of 397 bf16 relative-tolerance literals in test code (45\%) are at least
$1.6\times10^{-2}$, the \code{assert\_close} default. The artifact contains pinned commits, the detector and every
labeled diff.

\subsection{Three Ways a Fixed Tolerance Fails}\label{sec:three-failures}

\para{Permissive tests and missed bugs} vLLM's Mamba2 kernels process continuous-batch sequences in fixed-size chunks.
A wrong chunk-scan memory offset corrupted outputs whenever a sequence with an initial state started inside a
chunk longer than one tile. The test used $\mathit{atol} = \mathit{rtol} = 0.5$ except for chunk sizes 8 and 256,
compared only one output head and channel, and omitted bf16 continuous batching. The fix tightened tolerance to
$5\times10^{-3}$--$10^{-2}$ depending on sequence length~\cite{vllm21783}; a second fault on the same path survived
until a later release~\cite{vllm23279}. \Cref{sec:rq4} reports \sys's results on both.

\para{The reference is less accurate than the kernel} For a bf16 RMSNorm input-gradient element, Liger returns
0.00146484, the float64 reference approximately 0.00146491, and the bf16 PyTorch reference 0.0018158 (24\% off).
Comparing kernel with reference, \code{assert\_close} flags the correct kernel. PyTorch's defaults flag all 80 runs
in each bf16/fp16 RMSNorm and RoPE configuration we tested (\cref{sec:rq1}).

\para{Legitimate error depends on the input} The rounding error of a sum scales with $\sum|\text{terms}|$ while the
result scales with $|\sum\text{terms}|$, so a correct kernel's relative error is unbounded under cancellation.
Normalization kernels divide by row statistics, so a row whose mean square or variance is close to $\epsilon$
amplifies every rounding error. A tolerance tuned on Gaussian inputs therefore fails on heavy-tailed or degenerate
ones, and a tolerance wide enough for those is blind on Gaussian ones. \Cref{fig:tolerance-by-family} shows this for
the bf16 RMSNorm backward kernel: against the PyTorch reference, correct implementations need
$\mathit{atol} = \mathit{rtol} = 0.00625$ on Gaussian inputs but $0.125$ on near-degenerate rows, a 20-fold spread;
against the float64 reference the hardest family is a different one.

\begin{figure}[!htbp]
\centering
\includegraphics[width=0.62\linewidth]{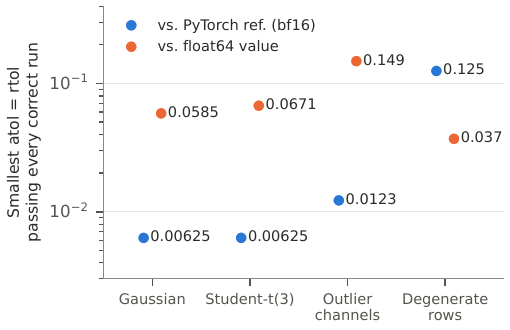}
\caption{Smallest $\mathit{atol} = \mathit{rtol}$ that passes every correct run of the bf16 RMSNorm backward kernel
($dX$ and $dW$), per input family. CPU emulations of correct kernels with different reduction orders, 20 seeds and
four shapes per family. The required tolerance changes by up to $20\times$ across families, and which family is
hardest depends on the reference.}\label{fig:tolerance-by-family}
\Description{Paired points compare the tolerances needed by four input families under two references. Gaussian and near-degenerate inputs require markedly different tolerances; changing the reference also changes the hardest family.}
\end{figure}

\subsection{Probabilistic Rounding-Error Analysis}\label{sec:prob-analysis}

For an inner product $s = \sum_i x_i y_i$ of length $n$ evaluated with unit roundoff $u$, the classical worst-case
bound is $|\hat s - s| \le \gamma_n \sum_i |x_i y_i|$ with $\gamma_n = nu/(1-nu)$, assuming $nu<1$ and
no overflow or underflow~\cite{higham2002accuracy}.
\citet{higham2019new} model the relative rounding errors as independent random variables of mean zero; later work
weakens independence to mean independence, under which each error has mean zero conditionally on the previous
ones~\cite{higham2020sharper,connolly2021stochastic}. Under the original model, with probability at least
$1 - 2n\exp(-\lambda^2(1-u)^2/2)$,
\begin{equation}
|\hat s - s| \le \tilde\gamma_n(\lambda)\sum_{i=1}^{n}|x_i y_i|,
\qquad \tilde\gamma_n(\lambda) = \exp\Big(\lambda\sqrt{n}\,u + \frac{n u^2}{1-u}\Big) - 1
= \lambda\sqrt{n}\,u + O(u^2)
\end{equation}
for fixed $n$ and $\lambda$ as $u \to 0$. Its $\sqrt{n}$ growth matches observed errors far better than $n$, and the model has been validated for random
data~\cite{higham2020sharper,ipsen2020inner}. \sys adopts the conditional mean-zero formulation but needs more than bounds for single
linear-algebra kernels: per-element bounds for whole tensor programs, dependence between elements, quantized formats
and real math libraries. \Cref{sec:calculus} builds that calculus.

\section{Problem Formulation}\label{sec:problem}

\subsection{References, Precision Plans and Correct Implementations}

A kernel $K$ maps input tensors $x$ to output tensors $o = K(x)$. Its specification is a reference program $\Rp$, a
straight-line tensor program over real arithmetic whose exact result on $x$ we write $v = \Rp(x)$. A precision plan
$\Pi$ annotates $\Rp$ with the decisions every implementation must respect:
\begin{itemize}
\item the format to which each stored intermediate and each output is rounded;
\item the accumulation format of each reduction and contraction;
\item each quantizer, with its format and scale granularity;
\item the math-library profile, i.e.\ the accuracy guaranteed for each elementary function.
\end{itemize}
The plan deliberately leaves out what an implementation is free to choose: the order and tree shape of every
reduction, its blocking into tiles, and fused multiply-add contraction. This freedom is exactly what makes bit-exact
testing impossible. Anything that rounds is not free: a split-K decomposition that stores partial sums in low
precision performs extra roundings, so it must appear in $\Pi$, and an implementation that adds roundings outside
$\Pi$ is outside the class.

\begin{definition}[Correct implementations]\label{def:class}
$\Cls$ is the set of implementations that evaluate $\Rp$ by performing, in any order allowed above, the
roundings that $\Pi$ prescribes (fused multiply-add may skip the rounding of a product), with elementary functions within the profile. Their outputs on $x$ form the
\emph{legal set} $L(x) = \{K(x) : K \in \Cls\}$.
\end{definition}

\para{Accuracy and parity} In \emph{accuracy mode}, $\Rp$ is the intended mathematics and $\Pi$ the kernel's
precision plan; a quantizer then contributes error relative to the unquantized input. In \emph{parity mode}, $\Rp$
includes selected rounding steps of a trusted implementation, such as a Hugging Face intermediate cast to bf16,
as semantic operations. It still admits the choices left open by $\Pi$ and does not require bit-exact replay. The two modes can
disagree. The GEGLU kernel that Liger-Kernel PR \#986 changed was more accurate than the Hugging Face reference, so
it passes our accuracy check. In parity mode the uncertain input to that semantic cast also makes our enclosure
too wide to distinguish the change (\cref{sec:rq4}).

\subsection{Sound and Sensitive Oracles}

An oracle $O$ maps an input and an output to \emph{pass} or \emph{fail}; \sys may also answer \emph{inconclusive} (A3).

\begin{definition}[$\delta$-soundness]\label{def:soundness}
$O$ is $\delta$-sound for $\Cls$ under a rounding model $M$, a probability law for the rounding errors, if, for every $K \in \Cls$ and every input $x$,
$\Pr_M[O(x, K(x)) = \mathit{fail}] \le \delta$. The probability is over the rounding errors as modeled by $M$.
\end{definition}

$\delta$-soundness is a statement about the model $M$. On real hardware a fixed input produces fixed rounding errors,
so the probability is not a frequency over repeated runs. We therefore state precisely which model
\cref{thm:soundness} uses, test it empirically on 33,960 GPU runs per GPU (\cref{sec:rq1}), and offer a fixed-input
mode whose guarantee needs no zero-mean rounding assumption.

We measure the \emph{sensitivity} of $O$ by its detection rate on faulty implementations over the evaluated inputs. A $\delta$-sound oracle
passes correct implementations except with probability $\delta$ under the model; the less it accepts beyond their
likely outputs, the more faults it can detect. We measure this \emph{tightness} on correct runs as $\max_i |o_i - v_i| / B_i$, the fraction of the
per-element bound $B$ that the worst element of a run uses.

\subsection{Rounding Model and Assumptions}\label{sec:assumptions}

\begin{itemize}[leftmargin=2.6em, labelwidth=2em, labelsep=0.6em, align=left]
\item[\textbf{A1}] \emph{(Rounding model.)} The model executes an implementation with each ordinary rounding of an
  intermediate value $z$ to a format $f$ replaced by $z + \xi$, where $\xi$ is a random variable with $|\xi|$ at most
  half the spacing of $f$ at $z$ (gradual underflow included) and $\mathbb{E}[\xi \mid \text{all earlier errors}] = 0$,
  in the event order of Appendix~A. Semantic casts, quantizer scales and each output's final store are $\fl_f$ of the
  model's values; both checks observe these stored outputs, which are not reused inside the call. Round-to-nearest is the path $\xi = \fl_f(z) - z$: it meets
  the bounds, but for a fixed input its errors are fixed numbers.
\item[\textbf{A2}] \emph{(Math library and accumulators.)} Each elementary-function call $\hat f$ satisfies
  $|\hat f(z) - f(z)| \le \LIB_f(z)$ for the active profile, and each hardware accumulator follows its calibrated
  model (\cref{sec:quantizers}).
\item[\textbf{A3}] \emph{(Domain.)} Every exact value is finite, modeled arithmetic and stored values do not overflow, and every divisor and every argument of a logarithm or root
  stays inside the function's domain over its whole interval $[v - \Ehp, v + \Ehp]$, so that every bound is finite. Otherwise \sys reports the element as
  uncertified and the result is inconclusive rather than a fault; \cref{sec:rq1} reports such outcomes under input shift.
\item[\textbf{A4}] \emph{(Implementation.)} \sys evaluates $v$ and $B$ in float64 rather than in real arithmetic.
  The checks use $\Bimpl = (1 + 2^{-15})B + 2^{-40}|v|$, assuming the float64 values $\hat v, \hat B$ satisfy
  $|\hat v - v| + |\hat B - B| \le 2^{-15}B + 2^{-40}|v|$.
\end{itemize}

A1 is the mean-independence model of \citet{higham2020sharper} and \citet{connolly2021stochastic}, a weakening of
the independence model of \citet{higham2019new}, and, like any rounding model, an idealization. For a fixed input,
deterministic round-to-nearest commits a fixed error, and even inputs drawn from a continuous distribution do not
make each error exactly mean-zero. A1 is plausible when test inputs randomize the low-order bits that rounding errors
depend on, and it fails on structured inputs, such as tensors of identical values, whose errors repeat.
\Cref{thm:soundness} is therefore a statement about the model; \cref{sec:evaluation} tests the model on real kernels,
and the fixed-input mode drops A1's zero-mean condition and uses deterministic worst-case bounds throughout. A2 is a hardware assumption
that calibration tests for the intrinsics and Triton's \code{tl.dot} on each GPU and compiler version but cannot prove
(\cref{sec:calibration}); for the cuBLAS GEMMs of one historical bug we assume \code{tl.dot}'s model without
calibrating cuBLAS. A1 assumes gradual underflow; flushing subnormals to zero is outside the model. A4 is tested,
not proved, against 100-digit arithmetic on sampled inputs of every shipped library reference and Liger's cross-entropy (\cref{sec:validation}), excluding SSD and the historical-bug references.

\subsection{Limits of Fixed Tolerances}\label{sec:fixed-tolerance}

Let $\rho(x)$ be the legal error radius of one output element, the largest $|K(x) - v(x)|$ over $K \in \Cls$. A
fixed tolerance accepts every legal output only if $\mathit{atol} + \mathit{rtol}\cdot|v(x)| \ge \rho(x)$ on every input the test will
ever see. For a sum, $\rho$ grows with $u\sum|\text{terms}|$ while $|v| = |\sum\text{terms}|$, so $\rho/|v|$ is
unbounded under cancellation, and $\rho$ grows with the input's scale and the reduction length. Any such
$(\mathit{atol}, \mathit{rtol})$ is therefore loose by the ratio between the hardest and easiest inputs it must
cover, often an order of magnitude (\cref{fig:tolerance-by-family}). An input-adaptive oracle avoids this trade-off by
using an input-specific region. Under A2--A4 and the bounded-error premise of A1, \sys's fixed-input mode gives a worst-case enclosure; its probabilistic mode
instead gives a model-conditional high-probability region, not a bound on every legal output.

\section{The \sys Error Calculus}\label{sec:calculus}

\subsection{Overview}

A mathematical shadow element carries the real-arithmetic reference value $v$ and an error model. The
implementation approximates these quantities in float64 under A4. Each operation propagates the reference value
and the error model according to the precision plan. The tester writes the
reference much like a PyTorch reference, passing the plan's formats as arguments (\cref{lst:rmsnorm}). At the end,
\sys turns each output's error model into a per-element bound and, for outputs stored in a low-precision format, into
an interval of acceptable representable values, and checks the kernel's output against both
(\cref{fig:workflow}).

\begin{figure}[!htbp]
\begin{lstlisting}[style=py]
import numoracle as nm
from numoracle import shadow as no

def ref_rms_norm_fwd(X, W, eps, dt="bf16"):          # Liger RMSNorm, casting_mode="llama"
    x, w = no.shadow(X), no.shadow(W)                  # inputs are taken as exact
    ss = no.sum(no.mul(x, x, "fp32"), -1, keepdim=True, acc="fp32", out=None)
    ms = no.div(ss, float(X.shape[-1]), "fp32", lib="triton_fast")    # fp32 div.full
    r = no.rsqrt(no.add(ms, eps, "fp32"), "fp32", lib="accurate")      # libdevice rsqrt
    xn = no.cast(no.mul(x, r, "fp32"), dt)             # normalized row stored in bf16
    return no.mul(xn, w, out=dt)                       # weight applied in fp32, stored in bf16

Y = liger_rms_norm(X, W, eps)                          # kernel under test
report = nm.check_kernel({"Y": (Y, ref_rms_norm_fwd(X, W, eps))}, delta=1e-6)
\end{lstlisting}
\caption{A \sys reference for the forward pass of Liger's RMSNorm, as used in our evaluation. Each operation names
the format it rounds to; \code{out=None} leaves a value unrounded, and \code{acc} sets the accumulation format of a
reduction.}\label{lst:rmsnorm}
\Description{Python listing constructs an RMSNorm shadow reference using explicit input, multiplication, reduction, division, reciprocal-square-root and cast operations, then compares the GPU output with that reference.}
\end{figure}

\begin{figure}[!htbp]
\centering
\begin{tikzpicture}[font=\small,
  box/.style={draw, rounded corners=2pt, align=center, inner sep=3pt, minimum height=9mm},
  data/.style={box, fill=black!5},
  arr/.style={-{Stealth[length=2mm]}, semithick}]
\node[data] (x) at (0,0) {inputs $x$};
\node[box] (k) at (2.75,0.85) {kernel $K$\\(Triton, GPU)};
\node[box] (r) at (2.75,-0.85) {reference $\Rp$ + plan $\Pi$\\shadow execution};
\node[data] (o) at (5.6,0.85) {output $o$};
\node[data] (b) at (5.6,-0.85) {$v,\,B$ and $v^{-},B^{-}$};
\node[box, fill=black!10, text width=3.95cm] (c) at (9.05,0)
  {\textbf{bound check}\\$|o - v| \le B$\\[3pt]\textbf{rounding-interval check}\\
   $\fl(v^{-}\!-\!B^{-}) \le o \le \fl(v^{-}\!+\!B^{-})$};
\node[data] (res) at (12.2,0) {pass\\inconclusive\\fail + witness};
\draw[arr] (x.north) |- (k.west); \draw[arr] (x.south) |- (r.west);
\draw[arr] (k) -- (o); \draw[arr] (r) -- (b);
\draw[arr] (o.east) -- (c.west |- o.east); \draw[arr] (b.east) -- (c.west |- b.east);
\draw[arr] (c) -- (res);
\end{tikzpicture}
\caption{The reference and the kernel see the same inputs. Shadow execution turns the precision plan into
per-element bounds, and the two checks compare the kernel's output with them.}\label{fig:workflow}
\Description{Inputs branch to a GPU kernel and a precision-annotated shadow reference. The kernel output and reference bounds feed a magnitude check and a rounding-interval check, producing pass, inconclusive, or fail with a witness.}
\end{figure}

\subsection{Shadow Elements}\label{sec:shadow}

A shadow element $(v,D,\mathcal C)$ carries the exact value $v$, a deterministic error-magnitude bound $D\ge0$, and
random error components. Each component $c$ has a key $S_c$ of tensor axes, sub-Gaussian proxy $P_c$, worst-case
bound $W_c$, rounding origins $\orig_c$, global incidence cap $I_c$, and complete coordinate-placement token
(\cref{sec:shape}). The element also records origins $\dorig$ whose effects were folded into $D$, for product
ordering (\cref{sec:products}). At nonterminal values, on the event $G$ of \cref{sec:soundness}, $\hat y=v+d+\sum_c r_c$, where $|d|\le D$, $|r_c|\le W_c$, and
$r_c$ is a martingale sum over $\orig_c$ with proxy at most $P_c$. The residual $d$ may be random.
Terminal-store summaries are numerical envelopes of the deterministic round-to-nearest values actually stored,
handled separately in the proof (\cref{sec:soundness}).

The key records axes along which events may be shared and errors must be combined coherently; their coefficients
need not agree. For each event, the positions where it can occur project to at most $I_c$ distinct coordinate tuples
outside $S_c$. Thus a line of length $n$ along $d\notin S_c$ contains at most $\min(I_c,n)$ occurrences of one event.
Two summaries drive everything else: the
worst-case total $\Wt$, and the probabilistic scale $R$, which combines components correlation-safely within groups
that share rounding operations and in quadrature across groups. The groups $g$ are the connected components of the
relation ``$\orig$ sets intersect''. The high-probability error bound of an intermediate value is $\Ehp$, with
$\lambda_I = 9$ (\cref{sec:soundness} accounts for its failure probability):
\begin{equation}\label{eq:summaries}
\Wt = \sum_{c} W_c, \qquad
R = \Big(\sum_{g}\Big(\sum_{c \in g}\sqrt{P_c}\Big)^{2}\Big)^{1/2}, \qquad
\Ehp = D + \min(\Wt,\ \lambda_I R).
\end{equation}
\Cref{fig:rules} states the core propagation rules; the following subsections explain them. Here $u_f$ is the unit roundoff and $\eta_f$ is half the smallest positive subnormal value of format $f$.

\begin{figure}[!htbp]
\footnotesize
\mprset{sep=1.4em}
\begin{mathpar}
\inferrule*[right=Round]
  {z \Downarrow (v, D, \mathcal{C}) \\ q = u_f(|v| + \Ehp) + \eta_f \\ e\ \text{fresh}}
  {\mathrm{rnd}_f(z) \Downarrow \big(v,\ D,\ \mathcal{C} \uplus \{\langle \emptyset, q^2, q, \{e\}\rangle\}\big)}

\inferrule*[right=Add]
  {x_1 \Downarrow (v_1, D_1, \mathcal{C}_1) \\ x_2 \Downarrow (v_2, D_2, \mathcal{C}_2)}
  {x_1 + x_2 \Downarrow (v_1 + v_2,\ D_1 + D_2,\ \mathcal{C}_1 \oplus \mathcal{C}_2)}

\inferrule*[right=Bcast]
  {x \Downarrow (v, D, \mathcal{C})}
  {\mathrm{bcast}_d(x) \Downarrow \big(v,\ D,\ \{\langle S \cup \{d\}, P, W, O\rangle : \langle S, P, W, O\rangle \in \mathcal{C}\}\big)}

\inferrule*[right=Sum]
  {x_i \Downarrow (v_i, D_i, \mathcal{C}_i)\ \ (i = 1..n) \\
   \langle S, P', W', O'\rangle = \textstyle\bigoplus^{d}_i |a_i|\cdot\langle S, P_i, W_i, O_i\rangle \\
   q_{\mathit{acc}}\ \text{as in}\ \eqref{eq:accumulation}}
  {\textstyle\sum^{d}_i a_i x_i \Downarrow \big(\sum_i a_i v_i,\ \sum_i |a_i| D_i + n\eta_{\mathit{acc}},\
   \{\langle S \setminus \{d\}, P', W', O'\rangle\} \uplus \{\langle \emptyset, n q_{\mathit{acc}}^2, n q_{\mathit{acc}}, E\rangle\}\big)}

\inferrule*[right=Mul-Ord]
  {X \Downarrow (v_X, D_X, \mathcal{C}_X) \\ Y \Downarrow (v_Y, D_Y, \mathcal{C}_Y) \\
   \mathrm{ev}(X) \prec \mathrm{ev}(Y) \\ \bar a = |v_X| + \Ehp(X)}
  {X \cdot Y \Downarrow \big(v_X v_Y,\ D_X|v_Y| + \bar a D_Y,\ |v_Y|\cdot\mathcal{C}_X \oplus \bar a\cdot\mathcal{C}_Y\big)}

\inferrule*[right=Mul]
  {X \Downarrow (v_X, D_X, \mathcal{C}_X) \\ Y \Downarrow (v_Y, D_Y, \mathcal{C}_Y) \\
   \mathrm{ev}(X) \not\prec \mathrm{ev}(Y)}
  {X \cdot Y \Downarrow \big(v_X v_Y,\ D_X|v_Y| + |v_X|D_Y + \Ehp(X)\Ehp(Y),\ |v_Y|\cdot\mathcal{C}_X \oplus |v_X|\cdot\mathcal{C}_Y\big)}

\inferrule*[right=Fun]
  {x \Downarrow (v, D, \mathcal{C}) \\ \theta = \tfrac12 \sup_{|t| \le \Ehp}|f''(v+t)|\,\Ehp^2}
  {f(x) \Downarrow \big(f(v),\ |f'(v)|\,D + \theta + \LIB_f,\ |f'(v)|\cdot\mathcal{C}\big)}

\inferrule*[right=Check]
  {y \Downarrow (v, D, \mathcal{C}) \\ |o - v| \le D + \min(\Wt, \lambda R) \\
   \fl_f(v^{-} - B^{-}) \le o \le \fl_f(v^{-} + B^{-})}
  {\vdash_\delta o \approx y}
\end{mathpar}
\caption{Core rules of the calculus. A component $\langle S, P, W, O\rangle$ has key $S$, proxy $P$, worst-case
bound $W$ and origin set $O$; incidence and placement metadata are omitted here and propagated as in
\cref{sec:shape}. A fresh private component has $I=1$.
$a\cdot\langle S, P, W, O\rangle = \langle S, a^2P, |a|W, O\rangle$. $\oplus$ merges
components with equal keys, adding $\sqrt{P}$ if their origins intersect and $P$ otherwise; $\bigoplus^{d}$ does the
same over the $n$ positions of axis $d$, adding coherently if $d \in S$ and as
$m_d\sum P_i$ otherwise, with $m_d=\min(I_c,n)$. $E$ is a fresh
set of accumulation events, $\mathrm{ev}(X) = \orig(X) \cup \dorig(X)$, and $\prec$ means that every event of one
set precedes every event of the other. In \textsc{Fun}, $\LIB_f$ is the library bound maximized over
$[v - \Ehp, v + \Ehp]$. Nonlinear rules also add $\orig$ of their inputs to $\dorig$ of what they
fold into $D$.}\label{fig:rules}
\Description{Inference rules specify input, rounding, addition, summation, multiplication, nonlinear functions and the final check. The caption defines component scaling, merging, shared-event incidence and ordering metadata.}
\end{figure}

\subsection{Rounding}\label{sec:rounding}

When the plan rounds an intermediate $z$ with shadow $(v, D, \mathcal{C})$ to format $f$, \sys adds a private
component, with empty key, incidence cap 1 and a fresh origin, of size
\begin{equation}\label{eq:rounding}
q = u_f\big(|v| + \Ehp\big) + \eta_f, \qquad P = q^{2}, \qquad W = q.
\end{equation}
The bound uses $|z| \le |v| + \Ehp$, which holds on the event analyzed in \cref{sec:soundness}.

\subsection{Linear Operations, Broadcasting and Reductions}\label{sec:linear}

Addition merges equal-key components coherently when origins intersect and in quadrature otherwise
(\cref{fig:rules}). Broadcasting along $d$ adds $d$ to each key. For reduction over $n$ positions of $d$, let
$a_i=1$ for a sum and $a_i=1/n$ for a mean. If $d$ is in a component's key, its $n$ terms may share events and add coherently,
$\sqrt{P'} = \sum_i |a_i|\sqrt{P_i}$. Otherwise each event occurs in at most $m_d=\min(I_c,n)$ terms, giving
$P' = m_d\sum_i a_i^2 P_i$ by Cauchy--Schwarz. Pure quadrature is the case $I_c=1$.

The reduction projects the
event's support onto the retained coordinates, so it preserves the inherited global cap (\cref{sec:shape}).

The accumulation's own rounding errors are bounded independently of the summation order. Whatever the tree, at most
$n$ additions occur, the computed addends $\hat t_i$ satisfy $\sum_i|\hat t_i| \le \hat S = \sum_i |a_i|(|v_i| + E_{\mathrm{hp},i})$, and
by induction every partial sum is at most $(1+u)^n(\hat S + n\eta) \le e^{nu}(\hat S + n\eta)$. \sys adds one private
component for the relative part and charges the underflow part to $D$, which needs no zero-mean assumption: many tiny
addends of one sign may all round the same way.
\begin{equation}\label{eq:accumulation}
q_{\mathit{acc}} = u_{\mathit{acc}}\, e^{n u_{\mathit{acc}}}\big(\hat S + n\,\eta_{\mathit{acc}}\big), \quad
P_{\mathit{acc}} = n\,q_{\mathit{acc}}^{2}, \quad W_{\mathit{acc}} = n\,q_{\mathit{acc}}, \quad
D \mathrel{+}= n\,\eta_{\mathit{acc}}.
\end{equation}
Unlike the classical factor $1/(1-nu)$, the factor $e^{nu}$ stays finite for every $n$; \sys warns when $nu \ge 1$,
where the bound is valid but loose.

Ignoring keys is not a harmless simplification. Multiplying 4,096 ones by the bf16 scalar $1 + 3\cdot2^{-9}$ and
summing gives an error of exactly 8: the scalar's rounding error, $2^{-9}$, is added 4,096 times. A model that treats
the 4,096 errors as independent bounds the error by 1.42 and raises a false alarm on every run. With the scalar's
component keyed by the broadcast axis, \sys's bound is 17.1.

\subsection{Products and Nonlinear Operations}\label{sec:products}

Products are where classical analyses add second-order terms. Summed worst-case over a contraction, those terms
can dominate FP8 and FP4 bounds. \sys avoids them with an exact decomposition. Let $\hat X = X + e_X$ and $\hat Y = Y + e_Y$. Then
\begin{equation}\label{eq:ordered}
\hat X\,\hat Y - X\,Y = (\hat X - X)\,Y + \hat X\,(\hat Y - Y).
\end{equation}
Suppose every rounding operation that affects $\hat X$ precedes every one that affects $\hat Y$ in program order, as
when two operands are quantized independently. ``Affects'' includes operations whose effect an earlier nonlinear step
folded into $D$: a row maximum computed from rounded values depends on those roundings even though its error is now
deterministic, and multiplying it by the same values is not an ordered product. \sys therefore tests ordering on
$\orig \cup \dorig$. Then the coefficient $\hat X$ of $e_Y$ is fixed before $e_Y$'s errors occur. Under A1, $e_Y$'s
components therefore remain martingale sums, with coefficients bounded by $|v_X| + \Ehp(X)$, and no second-order term
is needed (\cref{lem:ordered}). When origins are shared or interleaved, as in $XX^{\top}$, \sys propagates both
first-order terms with exact coefficients and adds the second-order term $\Ehp(X)\,\Ehp(Y)$ to $D$.

Matrix multiplication applies the same rule along the contraction axis and records contraction sharing. An error in
$X[i,k]$ reaches every column $j$ of the result, so $X$'s components gain the column axis in their key; $Y$'s gain
the row axis. An unshared contraction axis of length $K$ uses the factor $\min(I_c,K)$; projection of that axis
and sharing of the new axis preserve the inherited cap. Division uses first-order terms with exact coefficients and adds a bounded second-order remainder to
$D$.

An elementary function $f$ scales every component by $|f'(v)|$, a deterministic coefficient. It adds
$|f'(v)|\,D$, a Taylor remainder and the library error of the active profile to $D$:
\begin{equation}\label{eq:taylor}
\big|f(v+e) - f(v) - f'(v)\,e\big| \le \tfrac{1}{2}\sup_{|t|\le \Ehp}|f''(v+t)|\;\Ehp^{2}.
\end{equation}
Linearizing at the perturbed point instead, without a remainder, is unsound: $f'$ evaluated at a perturbed point
depends on the error itself, so the result is no longer zero-mean. (Charging the whole propagated error to $D$ would
be sound but loose.)

\subsection{Quantizers, Math Libraries and Accumulators}\label{sec:quantizers}

\para{Quantizers} In accuracy mode, quantizing $x$ with scale $s$ to a grid contributes, per element, at most half
the grid spacing at $|x|/s$ times $s$, plus a deterministic clipping error for values beyond the grid's maximum. \sys
supports per-tensor, per-row and per-group FP8 scales and block-scaled MXFP4 and NVFP4. Scales are computed from the
exact values and widened to cover any input within $\Ehp$, because an implementation computes its scale from its own,
perturbed values. In parity mode the quantizer belongs to the reference, and input uncertainty can change its output when the interval crosses grid boundaries. The same holds for a \emph{semantic cast}, a rounding that the reference itself performs: the reference
value becomes $\fl_f(v)$, the kernel computes $\fl_f(\hat y)$, and
$\fl_f(\hat y) - \fl_f(v) = (\hat y - v) + (\fl_f(\hat y) - \hat y) + (v - \fl_f(v))$. \sys keeps the components of
$\hat y - v$ and adds $h(|v| + \Ehp) + h(|v|)$ to $D$, where $h$ is half the format's spacing at a magnitude.

\para{Math libraries} Elementary functions are rarely correctly rounded on GPUs. \sys charges each call a
deterministic error $\LIB_f(z) = k_f\,u\,m_f(z) + \eta$, with an envelope $m_f$ chosen per function and profile
(\cref{tab:profiles}). The \emph{accurate} profile models CPU libraries and CUDA's libdevice. Its one exception is
GELU: standard implementations evaluate $0.5\,z\,(1 + \tanh(\cdot))$, which loses about $u|z|$ to cancellation for
negative $z$, so the envelope is $|z|$ rather than $|f(z)|$. The \emph{triton\_fast} profile models Triton's default
fp32 lowering, where the rounding of $z\log_2 e$ alone contributes a relative error of about $u|z|$ to $\exp$. Every
profile entry also adds a floor $\max(\eta, 2^{-96})$: a library may return zero for a tiny result, as fp32
$\mathrm{sigmoid}(-90)$ does when its internal $\exp$ overflows, and a purely relative envelope would reject that.
Both profiles are generous by design and are tested on each GPU (\cref{sec:calibration}).

\begin{table}[!htbp]
\centering
\caption{Library error bounds $\LIB_f(z)$ per profile ($u$: fp32 unit roundoff; $\sigma$: sigmoid). Every entry
except correctly rounded operations also adds $\max(\eta, 2^{-96})$.}\label{tab:profiles}
\footnotesize
\begin{tabular}{lll}
\toprule
Function & accurate & triton\_fast\\
\midrule
exp & $4u\,e^{z}$ & $u\,e^{z}(4 + 2|z|)$\\
sigmoid & $4u\,\sigma(z)$ & $u\,\sigma(z)(12 + 2|z|)$\\
silu & $4u\,|\mathrm{silu}(z)|$ & $u\,|\mathrm{silu}(z)|(13 + 2|z|)$\\
log & $4u\,|\log z|$ & $u\,(4|\log z| + 4)$\\
rsqrt, approximate sqrt & $4u\,|f(z)|$ & $4u\,|f(z)|$\\
\code{sqrt\_rn}, \code{div\_rn} & correctly rounded (ordinary rounding) & correctly rounded (ordinary rounding)\\
tanh, erf & $4u\,|f(z)|$ & via libdevice, as accurate\\
gelu (tanh or erf form) & $8u\,|z|$ & $8u\,|z|$\\
fp32 division & correctly rounded & $4u\,|\text{quotient}|$\\
\bottomrule
\end{tabular}
\end{table}

\para{Accumulators} Tensor cores may accumulate with truncation rather than
round-to-nearest~\cite{fasi2021tensor}, which violates A1's zero mean. For such accumulators \sys charges the whole
accumulation error to $D$. Following \citet{fasi2021tensor}, we model each block of $b$ products as aligning the products and running sum to
the largest term, truncating each, and then truncating the block sum; under this model, each of these
$K + 2\lceil K/b\rceil$ truncations loses less than $2u\hat S$ in the normal range, so the charge is
$(K + 2\lceil K/b\rceil)\cdot 2u\hat S + 2K\eta$, the last term for truncation below the normal range. Calibration
found a bias toward zero in Triton's \code{tl.dot} on both GPUs and fitted $b = 16$ on the RTX PRO 6000 and $b = 8$ on
the A100 (\cref{sec:calibration}).

\subsection{Shape Operations and Duplication}\label{sec:shape}

Shape operations can copy one rounding event to several positions: for example, \code{repeat}, gathers with repeated
indices, or combining a tensor with a rotated copy of itself, as in rotary embeddings. The global incidence cap $I_c$
bounds how many coordinate tuples outside $S_c$ can contain any one event. Reductions project these coordinates and
preserve the cap; duplication along unshared axes and row-dependent gathers may increase it. A placement token
certifies a common support envelope for shared events, independently of their coefficients and masks. Merges take
the larger cap when origins are disjoint or tokens agree, and sum the caps otherwise. Concatenation sums operand
caps when origins overlap because operands occupy different coordinate ranges. Appendix~A gives the updates and
proves the invariant.
The certified reduction factor is $m_d=\min(I_c,n)$.

\subsection{Checks}\label{sec:checks}

For one kernel call whose outputs have $N$ elements in total, \sys computes for every output element the bound
\begin{equation}\label{eq:bound}
B = D + \min(\Wt,\ \lambda R), \qquad \lambda = \sqrt{2\ln(8N/\delta)},
\end{equation}
and performs two checks. Sharing $N$ across all outputs makes the false-alarm budget hold per kernel call rather than
per tensor. Both checks use $\Bimpl = (1 + 2^{-15})B + 2^{-40}|v|$ (A4), and the interval endpoints below are
rounded outward.

\para{Bound check} Every element must satisfy $|o_i - v_i| \le B_i$.

\para{Rounding-interval check} For an output stored in format $f$, \sys keeps the shadow of the value before the
final rounding, with exact value $v^{-}$ and bound $B^{-}$. Because $\fl_f$ is monotone, a correct implementation's
stored value lies in the interval
\begin{equation}\label{eq:interval}
\fl_f(v^{-} - B^{-}) \;\le\; o \;\le\; \fl_f(v^{-} + B^{-}),
\end{equation}
and every representable value in that interval is the rounding of some value within $B^{-}$ of $v^{-}$
(\cref{lem:interval}). The ideal check accepts exactly those
values. When $B^{-}$ is small compared with the spacing of $f$, the interval holds the correctly rounded value and,
for elements near a rounding boundary only, its neighbor. A store that truncates instead of rounding then produces a
rejected value in a large fraction of elements (\cref{prop:truncation}). In general, a symmetric tolerance on $|o-v|$ cannot express this interval: its endpoints depend on the position of $v^{-}$ between representable numbers. Each failure
reports the float64 reference value and computed bound as a witness.

\section{Soundness, Sensitivity and the Single-Run Limit}\label{sec:soundness}

\subsection{Soundness}

\begin{theorem}[Soundness]\label{thm:soundness}
Let $K \in \Cls$, let $x$ be an input, let $0<\delta<1$, and let
$N_I \cdot 2e^{-\lambda_I^2/2} \le \delta/2$, where $N_I$ conservatively counts the call's propagation checkpoints
(Appendix~A). Under A1--A4, the probability under the rounding model that \sys's bound
check or rounding-interval check fails on any output element of the call $K(x)$ is at most $\delta$. In fixed-input
mode, where $B = D + \Wt$ and $\Ehp = D + \Wt$ throughout, neither check fails, and only A2--A4 and the error bound in A1 are needed.
\end{theorem}

The proof (Appendix~\ref{app:proofs}) orders the ordinary \emph{rounding events}, one per ordinary
rounding operation and element, by program order. Let $G$ be the event that all propagation checkpoints stay within their envelopes $E^{*} \le \Ehp$; copies made by shape operations then stay within their own $\Ehp$ and need no tail events (Appendix~A). The invariant holds for each nonterminal element on the event that all earlier propagation checkpoints stayed
within their envelopes: the element's error equals $d + \sum_c r_c$ with $|d| \le D$, where each $r_c = \sum_e
\gamma_e\xi_e$ sums rounding errors of the events in $\orig_c$ with coefficients fixed before $e$ occurs, $|\gamma_e|
\le \bar\gamma_e$ and $\sum_e \bar\gamma_e^2 q_e^2 \le P_c$. Merges follow from Minkowski's inequality
(\cref{lem:merge}), ordered products from judging order on $\orig \cup \dorig$ (\cref{lem:ordered}), and reduction
factors $\min(I_c,n)$ from the global incidence invariant and Cauchy--Schwarz (\cref{lem:mult}); remainders, library
errors, underflow and semantic casts go to $D$. Azuma--Hoeffding bounds the first propagation checkpoint leaving its
envelope by $N_I\cdot 2e^{-\lambda_I^2/2} \le \delta/2$ ($\lambda_I = 9$ allows about $9.7\times10^{10}$ checkpoints
at $\delta = 10^{-6}$; \sys checks the count), using predictable truncation rather than conditioning the martingale
on the success event. The final store is $\fl_f$ of the pre-store value and is not reused. Its numerical summary
adds a rounding allowance to $R$, without treating the deterministic store residual as centered. One pre-store tail event of
probability $e^{1/2}\delta/(4N)$ per output covers both checks: $\delta/2 + e^{1/2}\delta/4 < \delta$ in total.
Monotonicity places the stored value in the checked interval (\cref{lem:interval}).

\begin{lemma}[Exactness of the rounding interval]\label{lem:interval}
Let $a \le b$ be finite reals in the nonoverflowing domain of round-to-nearest $\fl_f$, with a fixed tie rule.
Then $\fl_f$ maps $[a, b]$ onto exactly the representable values in
$[\fl_f(a), \fl_f(b)]$.
\end{lemma}

\begin{lemma}[Combining components]\label{lem:merge}
Let $X = \sum_e \alpha_e\xi_e$ and $Y = \sum_e \beta_e\xi_e$ be martingale sums with predictable coefficients,
deterministic coefficient-magnitude bounds, and
proxies $P_X, P_Y$. Then $X + Y$ has proxy at most $(\sqrt{P_X} + \sqrt{P_Y})^2$, and at most $P_X + P_Y$ when no
event appears in both sums.
\end{lemma}

\begin{lemma}[Ordered products]\label{lem:ordered}
If every rounding event affecting $\hat X$, including those folded into its deterministic part, precedes every
rounding event affecting $\hat Y$, then $\hat X(\hat Y - Y)$ restricted to $Y$'s random part is a martingale sum
whose inherited coefficients are multiplied by the predictable $\hat X$, bounded by $|v_X| + \Ehp(X)$ on $G$.
\end{lemma}

\begin{lemma}[Multiplicity]\label{lem:mult}
If each rounding event appears in the components of at most $m$ positions along $d$, the quadrature sum over $d$, multiplied by
$m$, bounds the proxy of the reduction. For an unshared axis of length $n$, the global incidence invariant supplies
$m=\min(I_c,n)$ and remains valid after the axis is projected away.
\end{lemma}

\subsection{Tightness}\label{sec:tightness}

Soundness only needs the rules to over-approximate; usefulness needs them to lose little, and we give no general
guarantee of that. $P$ is a sub-Gaussian proxy built from worst-case increment bounds, so $\lambda R$ can exceed the true
error's standard deviation by far more than $\lambda$ when actual rounding errors are much smaller than their bounds,
and correlation-safe merges can add more. We therefore measure tightness. On 200 CPU trials per configuration of emulated correct kernels with different reduction
orders, the worst element of a run uses a median of 88\% of its bound for GEMM, 62\% for attention, 97--99\% for
normalizations, activations and RoPE, where the half-ulp of the final store dominates, 22--27\% for FP8, MXFP4 and
NVFP4 GEMMs in accuracy mode and 5\% for cross-entropy; the median share of the bound taken by the
deterministic part $D$ is at most 6\%. No trial
exceeded it.

\subsection{Single-Run Detection Limits}\label{sec:single-run}

\begin{proposition}[Single-run limit]\label{prop:single-run}
Let $O$ accept every output in the legal set $L(x)$. If a faulty implementation $K'$ produces an output in $L(x)$ on
input $x$, then $O$ passes $K'$ on $x$.
\end{proposition}

The statement is elementary, but it separates two reasons why a fault escapes. Either the oracle is loose, with an
acceptance region much larger than $L(x)$, or the fault is invisible on that input because its effect stays inside
$L(x)$. Tightness addresses only the first reason; the second requires different inputs or more observable outputs.

\Cref{prop:single-run} assumes an oracle that accepts every legal output. A $\delta$-sound oracle need not: \sys's
default bound covers a correct kernel's error with probability $1 - \delta$ under A1, not every point of $L(x)$, and
only its fixed-input bound contains $L(x)$. So we cannot claim that the faults below are invisible to every oracle;
we claim only that their effect lies within \sys's bound on the inputs tested. Misplacing $\epsilon$ in $\mathrm{rsqrt}(ms +
\epsilon)$ changes outputs by a relative $10^{-6}$ when the mean square $ms$ is of order one, far below bf16's rounding
scale and within the measured fp32 bounds. On rows with $ms \approx \epsilon$ it changes them by orders of magnitude, and there \sys detects
every instance (\cref{sec:rq3}). Two further faults alter only per-row losses inside a cross-entropy kernel whose test
observes their mean; with reduction set to \code{none}, \sys detects both (Appendix~\ref{app:evaluation-details}). A third adds one rounding
to the loss and stays within the bound in both settings.

One might hope to pool evidence across runs instead, since a fault repeats in every run while rounding noise does
not. Pooling does not come for free, because correct kernels also deviate systematically. In the GPU runs, the
correct Unsloth RoPE gradient deviated from the float64 reference value in the same direction in every run, and correct
normalization kernels showed a magnitude bias of the same sign in 56--64\% of runs, where errors symmetric about zero would
give 50\%. A pooled test must bound such biases a priori or it will flag correct kernels; we leave such tests to
future work.

\begin{proposition}[Sensitivity to rounding faults]\label{prop:truncation}
Consider an output element stored in a format $f$ of at most 24-bit precision with pre-rounding bound
$B^{-} \le \ulp_f(v^{-})/8$, and assume A4. Suppose a faulty kernel computes the same
pre-rounding value but rounds toward zero, that the position of $v^{-}$ between consecutive representable numbers is
uniformly distributed, and that the pre-rounding value lies within $B^{-}$ of $v^{-}$ except with probability at most
$\delta$ (\cref{thm:soundness}). Then the implemented rounding-interval check rejects the element with probability at
least $1/4 - 10^{-4} - \delta$.
\end{proposition}

\begin{proof}[Proof sketch]
Let $\varphi$ be the fractional position of $|v^{-}|$ in ulps. If $5/8 < \varphi < 7/8$, every pre-rounding value
within $B^{-}$ of $v^{-}$ truncates to the lower neighbor, while $\fl_f(v^{-} - B^{-})$ is the upper one. Under A4,
float64 and $\Bimpl$ move the implemented lower end by less than $10^{-4}\,\ulp$ (Appendix~A).
\end{proof}

Over $N$ elements with independent positions, and since \cref{thm:soundness} bounds the \emph{joint} event that any
pre-rounding value of the call leaves its bound by $\delta$, a truncating store escapes with probability at most
$(3/4 + 10^{-4})^N + \delta$. This is why the rounding-interval check catches truncation faults that
the fixed tolerances used in practice miss (\cref{sec:rq2}).

\section{Implementation}\label{sec:implementation}

\subsection{Library and References}

\sys is a pure-PyTorch core library of 1,535 lines of code, excluding comments, blank lines and docstrings, with no other
dependencies. Shadow tensors store $v$, $D$, $P$ and $W$ in float64, with keys, origins and incidence/placement
metadata; variance estimates are diagnostic only. The library provides elementwise arithmetic
and casts; sums, means and cumulative sums; \code{matmul}, \code{bmm} and two-operand \code{einsum}; softmax,
including the flash-attention variant that stores probabilities in bf16; FP8 quantizers per tensor, row and group,
and MXFP4 and NVFP4 block quantizers; the elementary functions of \cref{tab:profiles}; and shape operations. Format
emulation, round-to-nearest and round-toward-zero for every format of \cref{sec:formats}, is exact in float64.

We wrote references for GEMM in bf16 and fp16, scaled FP8 GEMM, NVFP4 and MXFP4 GEMM, RMSNorm, LayerNorm, fused
RMSNorm with dynamic FP8 quantization, softmax, attention with flash numerics, GEGLU and SwiGLU forward and backward,
cross-entropy, rotary embeddings, and the Mamba2 chunked state-space scan. The Liger and Unsloth subjects share these
references, adapted to each kernel's precision plan (\cref{sec:rq5}).

\subsection{Calibrating Math Libraries and Accumulators}\label{sec:calibration}

\Cref{thm:soundness} relies on A2, so the profile must match the hardware and compiler. Before testing on a new GPU
or compiler version, a calibration step of under one minute evaluates every math intrinsic that the kernels under
test call on $2^{22}$ inputs, including arguments up to $|x| = 100$ where $\exp$ overflows, against float64. It stops
the campaign if any error exceeds its profile, or if an operation documented as correctly rounded misses half an ulp.
It also measures the signed accumulation error of \code{tl.dot} on positive inputs for $K$ from 256 to 16,384, tests
for bias in both directions, and fits the truncation period $b$. On the RTX PRO 6000 (sm\_120) and the A100 (sm\_80)
with Triton 3.6, every intrinsic stayed within its profile, the largest ratio being 0.80 of the bound (libdevice
\code{tanh}); \code{sqrt\_rn} and \code{div\_rn} matched correct rounding on the sampled inputs. \code{tl.dot} accumulation was biased toward
zero on both GPUs, with no positive bias, and consistent with $b = 16$ and $b = 8$ respectively. Calibration is
evidence for A2, not a proof: a sampled maximum is not a supremum. An adversarial A100 stress test, with one large product and many below $2^{-24}$ of
it, showed the small ones dropped, exceeding an earlier per-block charge up to 1.97 times; the alignment charge
(\cref{sec:quantizers}) covers all 240 configurations; a 32-case grid passed on both GPUs ($\le 20\%$ of the bound). Calibration also showed why such evidence is needed: on the
CPU, PyTorch's float32 \code{sqrt} missed correct rounding on 17\% of 65,544 samples, by up to 1.5 half-ulps.

\subsection{Validating the Calculus Against Witnesses}\label{sec:validation}

We developed the calculus test-first, turning inputs on which an emulated correct implementation exceeded a rule's
bound into regression tests. Initial witnesses led to the rounding-interval, sharing and ordered-product rules;
proof review exposed accumulation/underflow, lost provenance, semantic-cast and math-library-floor errors.
Later audits corrected shared-event duplication through \code{where} and coordinate shifts, quantizer half-spacing
after division, nonfinite and overflow handling, softmax shifts of perturbed inputs, and incidence counts under
composed shape operations. Appendix~\ref{app:evaluation-details} records the individual witnesses; \cref{sec:setup} describes
reference comparisons and GPU reruns after the repairs.

The revised rule set is checked by the CPU regression suite recorded in the artifact, including composed movement and reduction witnesses for the global-incidence invariant. On 200 trials per configuration of 21 operator
configurations, emulated correct kernels with different reduction orders never exceeded their bound. With the final implementation, we re-evaluated $v$ and $B$ of 40 outputs on 20 inputs each with 100-digit arithmetic, covering every shipped library reference and Liger's cross-entropy. The largest discrepancies were $2.94\times10^{-5}$ and $1.09\times10^{-9}$ of the A4 allowance for $v$ and $B$, respectively. This check excludes SSD and the historical-bug references.

\section{Evaluation}\label{sec:evaluation}

We answer six research questions:
\begin{itemize}[leftmargin=3.4em, labelwidth=2.8em, labelsep=0.6em, align=left]
\item[\textbf{RQ1}] Does \sys accept correct kernels and variants across GPUs and input shifts?
\item[\textbf{RQ2}] How many seeded faults does it detect compared with fixed oracles?
\item[\textbf{RQ3}] How does it compare with tuned tolerances, including under input shift?
\item[\textbf{RQ4}] Does it detect historical numerical bugs and pass their fixes?
\item[\textbf{RQ5}] What do writing references and running the checks cost?
\item[\textbf{RQ6}] Which calculus refinements matter, and how does $\delta$ affect the results?
\end{itemize}

\subsection{Setup}\label{sec:setup}

\para{Subjects} We test 13 Triton operators from three libraries: 11 Liger and Unsloth operators, forward and backward, and two vLLM
kernels, forward only. Liger-Kernel at commit
91e51df: RMSNorm (casting mode \code{llama}), LayerNorm (affine), SwiGLU, GEGLU, cross-entropy and rotary embeddings.
Unsloth at commit d68515a: RMSNorm, rotary embeddings, SwiGLU, GEGLU and cross-entropy. vLLM's Mamba2 chunked
state-space kernels, batched and with continuous batching (\code{ssd}, \code{ssd\_varlen}), at the merge commit
46ae7f6 of PR \#21783. Every operator runs in bf16, fp16 and fp32 on four shapes, including non-powers of two such as
$37\times1000$ and $13\times3000$ and vocabularies of 12,345 and 32,000. Cross-entropy runs with reduction
\code{mean} and \code{none}.

\para{Mutants} We apply seven mutation operators to the Triton source, one edit per mutant (\cref{tab:mutants}). We
also generate behavior-preserving variants that change block sizes, warps or pipeline stages, swap commutative
operands, re-associate sums or split reduction loops; any alarm on them is a false alarm. Mutants that crash on every
run or are bit-identical to the correct kernel on every run (inert) are excluded, which leaves 53 effective mutants
for Liger (of 79; 2 crashing, 24 inert), 27 for Unsloth (of 48) and 44 for vLLM (of 68).

\begin{table}[!htbp]
\centering
\caption{Mutation operators and effective mutants per library.}\label{tab:mutants}
\small
\begin{tabular}{llrrr}
\toprule
Operator & Change & Liger & Unsloth & vLLM\\
\midrule
M1 & drop an fp32 upcast on load & 3 & 2 & 0\\
M2 & insert an extra rounding to the input dtype & 12 & 8 & 8\\
M3 & downgrade an fp32 accumulator & 0 & 0 & 9\\
M4 & store with truncation instead of rounding & 11 & 2 & 4\\
M5 & misplace $\epsilon$ in $\mathrm{rsqrt}(x + \epsilon)$ & 2 & 1 & 0\\
M6 & shift a tail mask by one & 14 & 9 & 18\\
M7 & apply a constant or scale in low precision & 11 & 5 & 5\\
\midrule
Total & & 53 & 27 & 44\\
\bottomrule
\end{tabular}
\end{table}

\para{Inputs} Each configuration runs 20 seeds that cycle through three input families: Gaussian, Student-$t$ with
3 degrees of freedom, and Gaussian with 1\% outlier channels scaled by 20. For the normalization operators, a
degenerate-row phase draws rows whose mean square or variance is $\epsilon f$ with $f$ log-uniform in $[1/4, 4]$,
over 10 seeds. For vLLM, a probe phase uses layouts in which a sequence starts inside one chunk and continues into the
next. These probes are a separate historical-bug phase and do not enter the 480 main correct-kernel runs;
their corrected control includes the later PR~\#23279 fix.

\para{Oracles} We compare \sys's certified checks ($\delta = 10^{-6}$, per kernel call) with four fixed oracles: the
tolerances in each library's own tests at the pinned commit, checked by hand (file and line in the artifact; dtypes a
library does not test are left out); PyTorch's \code{assert\_close} defaults; $\mathit{atol} = \mathit{rtol} =
10^{-2}$; and the aiter formula with $f = 1$. The fixed oracles compare against the PyTorch reference in the kernel's
dtype, as the tests they come from do. RQ3 adds tuned tolerances, defined there.

\para{Metrics} A false alarm is a flagged run of a correct kernel or a behavior-preserving variant. A mutant is
killed by an oracle if, on the same input, the oracle fails the mutant and passes the correct kernel. Inconclusive
outcomes count as neither detections nor false alarms. We report counts and compare kill sets with exact McNemar
tests at the mutant level; these tests are exploratory because mutants share kernels.

\para{Platform and versions} We use an RTX PRO 6000 Blackwell (sm\_120) and A100-SXM4-40GB (sm\_80), with
PyTorch 2.11.0 (CUDA 12.8) and Triton 3.6.0. Results are for the RTX PRO 6000 unless stated otherwise.
Main campaigns and timings used earlier implementations. After the repairs in \cref{sec:validation}, we reran
LayerNorm, SSD and three historical bugs on both GPUs; all 10,290 correct/variant controls per GPU pass, with
LayerNorm retaining 10/11 main and 11/11 degenerate-row kills and SSD 41/44. RQ4 reports the historical reruns;
RQ6 uses new Liger/vLLM and original Unsloth ablations. The experimental ``final implementation'' is \code{e447e916};
the artifact adds guard patch \code{de63b01}, compared offline without GPU reruns (Appendix~\ref{app:evaluation-details}).
A100's Unsloth count is 26/28 instead of 25/27: one RoPE truncation mutant changes outputs only there.

\para{Use of AI assistants} Claude (Anthropic) and Codex (OpenAI) assisted with library/reference code,
mutation and analysis scripts, proof drafts, counterexample search and manuscript preparation. Codex produced the initial labels in \cref{sec:testing-today}; an author reviewed all 80 judgments against the diffs. Numerical measurements come from execution
records, with regression witnesses and high-precision checks in \cref{sec:validation}. The authors are responsible
for the claims and artifact.

\subsection{RQ1: False Alarms}\label{sec:rq1}

\sys raised no false alarm in any of the 3,960 runs of correct kernels or the 30,000 runs of behavior-preserving
variants, on either GPU (\cref{tab:false-alarms}). These trials test the reference and hardware contracts;
shared inputs preclude an independent-trial confidence bound, and the counts do not establish the model-level
$\delta=10^{-6}$ guarantee (\cref{thm:soundness}). The fixed oracles that compare against the PyTorch reference do
not share this property. \code{assert\_close} flags 24.0\% of correct Liger runs, including all 80 runs of the bf16
and fp16 RMSNorm and rotary-embedding configurations separately, and every run of the correct vLLM kernels, because the low-precision
reference is less accurate than the kernel (\cref{sec:three-failures}). The aiter formula flags 54 of 80 correct fp32
Liger SwiGLU runs, and 97\% of correct vLLM runs.

\begin{table}[!htbp]
\centering
\caption{False alarms on correct kernels (flagged runs; RTX PRO 6000; identical counts for \sys on the A100).
Degenerate rows: Liger and Unsloth normalization kernels. Library: only dtypes the library's tests cover.}
\label{tab:false-alarms}
\small
\begin{tabular}{lllll}
\toprule
Oracle & Liger & Unsloth & vLLM & Degenerate rows\\
\midrule
\sys & 0 / 1,680 & 0 / 1,440 & 0 / 480 & 0 / 360\\
Library tolerances & 58 / 1,280 (4.5\%) & 0 / 560 & 125 / 400 (31.2\%) & 41 / 280 (14.6\%)\\
$\mathit{atol} = \mathit{rtol} = 10^{-2}$ & 3 / 1,680 (0.2\%) & 2 / 1,440 (0.1\%) & 203 / 480 (42.3\%) & 86 / 360 (23.9\%)\\
aiter formula & 56 / 1,680 (3.3\%) & 125 / 1,440 (8.7\%) & 465 / 480 (96.9\%) & 0 / 360\\
\code{assert\_close} defaults & 404 / 1,680 (24.0\%) & 326 / 1,440 (22.6\%) & 480 / 480 (100\%) & 168 / 360 (46.7\%)\\
\bottomrule
\end{tabular}
\end{table}

On behavior-preserving variants the picture is the same: \sys flags none of 30,000 runs, while
\code{assert\_close} flags 23\%, 28\% and 100\% of the Liger, Unsloth and vLLM ones on the main inputs. The A100 reproduces every count
in \cref{tab:false-alarms} except one more library and one more \code{assert\_close} alarm on Liger's degenerate
rows, although its \code{tl.dot} truncates every 8 products instead of 16; calibration adapts the bounds.

\para{Distribution shift} We then ran every correct kernel and a sample of variants, 20 seeds per configuration, on
four shifted input families not used for tuning: heavy-tailed Student-$t_3$ inputs scaled to median $|x| = 1$ (S1;
the main runs use the same family scaled to unit variance), magnitudes spread
log-uniformly over $10^{\pm3}$ (S2), rows with a large common offset, $10^3 + \mathcal{N}(0,1)$, that make
normalizations cancel (S3), and magnitudes near the underflow threshold (S4). Inputs whose exact output would overflow were scaled
down just enough (320 of 3,600 cases under each of S2 and S3, all fp16 gated activations). Fixed oracles retain
their settings (\cref{tab:stress}, A100; RTX PRO 6000 differs in three Liger baseline counts under S4). The original GPU
campaigns recorded no false alarm among 14,400 correct and 2,616 variant runs per GPU. With the final implementation,
we reran the 536 correct calls and 98 equivalent-variant calls per GPU identified by incomplete A3 certificates:
all 634 were inconclusive and none failed. Terminal output error envelopes cross the fp16 range in these rescaled
GEGLU/SwiGLU cases; GEGLU's global bound on $|\tanh''|$ is loose at large arguments. Appendix~\ref{app:evaluation-details} gives
the rerun scope and reconstruction limits. Tuned
tolerances flag 29--30\% of correct runs under S2/S3 and library tolerances up to 39\%. Tiny magnitudes (S4) are
benign for them because absolute tolerance dominates, but a fixed $\mathit{atol}$ is then blind to relative errors
of any size. On 216 changed-output runs of five additional shifted-input faults, \sys detects 179 (all five faults)
versus 49 for $\mathit{atol}=\mathit{rtol}=10^{-2}$ (four faults); 11 of 144 correct inputs are inconclusive.

\begin{table}[!htbp]
\centering
\caption{False alarms on correct kernels under distribution shift (A100; Liger, Unsloth and vLLM pooled; 3,600 runs per
family, 2,240 for library tolerances). Tuned: the per-library tolerance with at most 1\% false alarms on every input
group of \cref{sec:rq3}, fixed before these runs. Parentheses: inconclusive correct calls in the affected-subset GPU reruns with the final implementation.}\label{tab:stress}
\small
\begin{tabular}{lcccc}
\toprule
Oracle & S1 heavy tail & S2 wide range & S3 large offset & S4 tiny magnitudes\\
\midrule
\sys & 0 & 0 (266) & 0 (270) & 0\\
Tuned tolerance & 2.6\% & 29.1\% & 30.1\% & 0.1\%\\
Library tolerances & 20.1\% & 25.6\% & 39.4\% & 0.2\%\\
$\mathit{atol} = \mathit{rtol} = 10^{-2}$ & 8.8\% & 22.2\% & 22.2\% & 1.4\%\\
aiter formula & 16.8\% & 19.8\% & 13.2\% & 17.2\%\\
\code{assert\_close} defaults & 34.9\% & 31.5\% & 36.3\% & 17.6\%\\
\bottomrule
\end{tabular}
\end{table}

\subsection{RQ2: Fault Sensitivity}\label{sec:rq2}

On the main inputs, \sys killed 116 of 124 effective mutants: 50 of 53 in Liger (94\%), 25 of 27 in Unsloth and 41 of
44 in vLLM (\cref{tab:kills}), against 59--83 for the fixed oracles. Leaving out any one of the 13 operators, it still
kills at least 23 more than the best fixed oracle, and it kills at least as many as each fixed oracle in every mutation
class except $\epsilon$ misplacement (M5). Restricted to the inputs that the libraries' own tests cover, it kills 106
against their 77. $p$-values are exploratory: mutants share kernels and come from a deterministic generator, not
a sample of real faults.

\begin{table}[!hbp]
\centering
\caption{Mutants killed on the main inputs and exact McNemar $p$ against \sys (mutant level; exploratory).}\label{tab:kills}
\small
\begin{tabular}{lrlrlrl}
\toprule
Oracle & Liger (of 53) & $p$ & Unsloth (of 27) & $p$ & vLLM (of 44) & $p$\\
\midrule
\sys & 50 & -- & 25 & -- & 41 & --\\
\code{assert\_close} defaults & 40 & 0.013 & 19 & 0.07 & 0 & $9\times10^{-13}$\\
aiter formula & 30 & $2\times10^{-6}$ & 21 & 0.12 & 32 & 0.022\\
Library tolerances & 30 & $4\times10^{-5}$ & 13 & 0.002 & 34 & 0.039\\
$\mathit{atol} = \mathit{rtol} = 10^{-2}$ & 26 & $8\times10^{-7}$ & 19 & 0.031 & 31 & 0.006\\
\bottomrule
\end{tabular}
\end{table}

Truncating stores (M4) separate the oracles most clearly: \sys kills all 11 Liger M4 mutants and the best fixed oracle kills 3. The bound check alone finds each; the rounding-interval check flags additional runs (\cref{sec:rq6}). Appendix~\ref{app:evaluation-details} gives exclusive kills and precision-specific run shares.

\para{Missed mutants} The eight misses expose three limits: loose fp32 library allowances, mutations that do not reduce accuracy, and errors indistinguishable under our single-run checks. They include the three epsilon misplacements, an extra rounding of Liger's loss, an Unsloth RoPE change, and three vLLM changes. Appendix~\ref{app:evaluation-details} gives per-case error and bound measurements; these are retained in the reported denominator.

\subsection{RQ3: Tolerance Tuning}\label{sec:rq3}

We give fixed tolerances strong advantages. A \emph{hindsight-optimal} tolerance picks, for each operator and dtype,
the $(\mathit{atol}, \mathit{rtol})$ grid point that kills the most mutants among points with no false alarm on any
correct main run. It is tuned on the test inputs and the mutants themselves, so no point of our grid kills more mutants
without a false alarm on these runs. A \emph{diagonal} tolerance sets $\mathit{atol} = \mathit{rtol}$ to the smallest value with no false
alarm on the correct main runs, using no mutant information. A \emph{fair} tolerance sets
$\mathit{atol} = \mathit{rtol}$ on a separate tuning phase of 10 Gaussian seeds on the first shape, as a developer
would before releasing a test. Each is evaluated against the PyTorch reference and against the float64 reference.

\begin{table}[!htbp]
\centering
\caption{Tuned tolerances on Liger (killed of 53 on the main inputs; false alarms on correct runs). Held-out: main
runs whose input family is not Gaussian and whose shape is not the tuning shape. Degenerate rows: kills of the 22
mutants with an effect there. Hindsight and diagonal tolerances were tuned on all main runs, so their held-out
column is in-sample.}\label{tab:tuned}
\footnotesize
\begin{tabular}{llrlllr}
\toprule
& & & \multicolumn{3}{c}{False alarms} & Degen.\\
\cmidrule(lr){4-6}
Tolerance & Reference & Killed & Main & Held-out & Degenerate rows & kills\\
\midrule
hindsight-optimal & PyTorch & 52 & 0 (tuned) & 0 & 167 / 240 (70\%) & 11\\
hindsight-optimal & float64 & 51 & 0 (tuned) & 0 & 125 / 240 (52\%) & 19\\
diagonal & PyTorch & 47 & 0 (tuned) & 0 & 123 / 240 (51\%) & 17\\
diagonal & float64 & 44 & 0 (tuned) & 0 & 117 / 240 (49\%) & 14\\
fair & PyTorch & 53 & 335 / 1,680 (19.9\%) & 202 / 819 (24.7\%) & 210 / 240 (88\%) & 17\\
fair & float64 & 44 & 302 / 1,680 (18.0\%) & 190 / 819 (23.2\%) & 151 / 240 (63\%) & 9\\
\sys & float64 & 50 & 0 / 1,680 & 0 / 819 & 0 / 240 & 22\\
\bottomrule
\end{tabular}
\end{table}

Three findings stand out (\cref{tab:tuned}). \emph{On the inputs it was tuned on, a hindsight tolerance matches
\sys}: it kills 52 mutants against 50 (McNemar $p = 0.62$), but it was chosen with access to every correct run and
every mutant. \emph{Without hindsight, the trade-off returns}: tolerances tuned on Gaussian inputs of one shape flag
18--20\% of correct main runs and 23--25\% of held-out runs. \emph{Under input shift, all six tuned baselines raise false alarms}: on
degenerate rows, tolerances that never false-alarmed on the main inputs flag 49--70\% of correct runs and the fair ones
63--88\%, while \sys flags none and kills all 22 mutants with an effect there, against 9--19, including every instance
of both $\epsilon$ misplacements.

A single $\mathit{atol} = \mathit{rtol}$ per operator and dtype may be too weak a baseline, so we strengthened it. We
tuned 400 configurations on the tuning phase only: one tolerance per operator or per output (loss, input and weight
gradients separately), $\mathit{atol} = \mathit{rtol}$ or the full frontier of a two-dimensional
$(\mathit{atol}, \mathit{rtol})$ grid, against the PyTorch or the float64 reference, each widened by safety factors 1,
2, 4, 8 and 16. \Cref{fig:frontier} plots each configuration's false-alarm rate against
its kills; the lines show, in hindsight and per panel, the best at each false-alarm rate.

No configuration reaches \sys's corner. At most 1\% false alarms on every Liger input group permits at most
31 of 53 main-mutant kills and, separately, 16 of 22 degenerate-row kills.

\sys kills 50 and 22, respectively
(McNemar $p = 7\times10^{-5}$ on the main mutants). Allowing 10\% false alarms raises the tuned counts to 42 and 20
($p = 0.04$). Configurations that kill 50 main mutants flag at least
12\% of correct main runs and at least 35\% of degenerate-row runs. The same holds on the other libraries. For Unsloth, at most
1\% false alarms everywhere allows at most 22 of 27 kills, against \sys's 25 ($p = 0.38$); for vLLM, at most 1\% on
main and held-out inputs allows 36 of 44, against 41 ($p = 0.13$), reaching 41 costs 18\% false alarms.

\begin{figure}[!hbp]
\centering
\includegraphics[width=\linewidth]{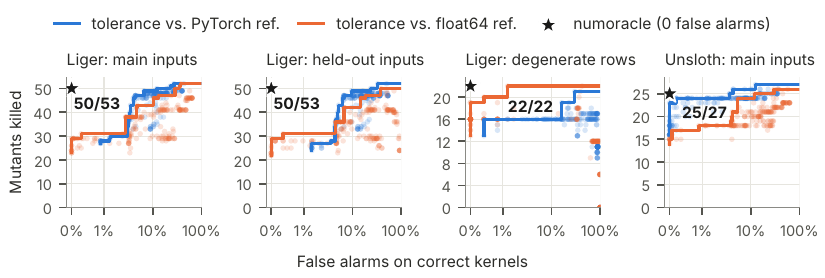}
\caption{False alarms versus mutants killed for 400 tuned tolerance configurations (dots) and the best one at each
false-alarm rate (lines), against \sys ($\star$). Liger panels: main inputs, held-out family and shape, degenerate
rows; right: Unsloth main inputs. The $x$ axis is linear up to 1\% and logarithmic beyond.}\label{fig:frontier}
\Description{Four panels plot false-alarm rate against mutants killed for tuned tolerances, with a star marking numoracle. Panels cover three Liger input groups and Unsloth main inputs. The plotted tolerance configurations do not reach the star.}
\end{figure}

\para{Percentile calibration} Setting tolerances to 1.5 times the 95th percentile of earlier correct-run errors~\cite{sarkar2026calibration} flags 32\%, 31\% and 41\% of later correct Liger, Unsloth and vLLM runs against PyTorch, and 90--100\% on degenerate rows. Appendix~\ref{app:evaluation-details} reports the per-operator, float64-reference and evaluation-set calibration variants (23--54\% false alarms across the tuning-phase variants).

\subsection{RQ4: Historical Bugs}\label{sec:rq4}

We reproduced seven historical bugs from Liger-Kernel and Unsloth~\cite{ligerprs,unslothprs} by running the source before and after each fix on
the same inputs, 20 seeds per input family, against an independent float64 reference (\cref{tab:historical}). An
input counts as a detection if the pre-fix version fails and the post-fix version passes. We compare with the projects'
own regression tests and with tuned tolerances: fair ones calibrated on Gaussian seeds disjoint from the evaluation
seeds, and hindsight ones chosen with the pre-fix outcomes in view. Input families on which a bug does not show are kept in the counts; the two negative
control families of \#10616, where the faulty path is not taken, are not (\sys passes both versions on all 80 control
runs; the fair PyTorch-referenced tolerance flags both in 20).

\sys raised no false alarm on any post-fix version and detected six of the seven bugs. Hindsight float64 tolerances also raise no post-fix alarm: on the RTX PRO 6000 they detect one more input for \#1433, match \sys for \#1498, \#10616 and \#10617, and detect fewer for \#1329 and \#92. The fair float64 tolerance also matches \#1498 without an alarm. Hindsight choices use the pre-fix outcomes; they are not available to a developer before seeing those faults. Where a project's pre-fix test could be replayed, it missed every bug except
\#1329 on high-confidence logits. For \#1433, \sys detected 128 of 200 inputs where the fair PyTorch-referenced
tolerance detected 1 and flagged the post-fix version 101 times. For \#1329 it detected every high-confidence input in
bf16 and fp16 and 9 of 20 Gaussian fp16 inputs; in fixed-input mode, which drops A1's zero-mean condition, the wider deterministic bound no
longer detects the fp16 high-confidence cases (0 of 20).

\begin{table}[!htbp]
\centering
\caption{Historical bugs in accuracy mode: paired detections (pre-fix fails and fixed passes on the same input), with false alarms
on the post-fix version in parentheses (RTX PRO 6000). Rows \#1433, \#1498 and \#10616 use the reruns described in \cref{sec:setup}. Project-test tolerances are replayed only for supported dtypes (bf16 here). Fair and hindsight tolerances: tuned as in RQ3.}\label{tab:historical}
\footnotesize
\setlength{\tabcolsep}{2.5pt}
\begin{tabular}{>{\raggedright\arraybackslash}p{1.9cm}>{\raggedright\arraybackslash}p{4.15cm}lllll}
\toprule
& & & Pre-fix & Fair, & Fair, & Hindsight,\\
Bug & Fault & \sys & project test & PyTorch ref. & float64 ref. & float64 ref.\\
\midrule
Liger \#1329 & cross-entropy: target-class gradient after an early low-precision store & 49/80 (0) & 20/40 (0) & 46/80 (3) & 39/80 (1) & 40/80 (0)\\
Liger \#1433 & fused linear + JSD: projection kept in the input precision & 128/200 (0) & 0/120 (40) & 1/200 (101) & 65/200 (82) & 129/200 (0)\\
Liger \#1498 & LayerNorm: mean and rstd saved in the input precision & 40/40 (0) & 0/20 (0) & 35/40 (5) & 40/40 (0) & 40/40 (0)\\
Liger \#92 & RMSNorm: fp32 casts at the reference's positions & 149/160 (0) & 0/80 (0) & 152/160 (8) & 91/160 (16) & 110/160 (0)\\
Liger \#986 & GEGLU: intermediate not cast as Hugging Face does & 0/80 (0) & -- & 74/80 (6) & 0/80 (19) & 0/80 (0)\\
Unsloth \#10616 & cross-entropy with softcap: padding mask not restored & 120/120 (0) & -- & 83/120 (37) & 118/120 (2) & 120/120 (0)\\
Unsloth \#10617 & RMSNorm backward with strided gradients & 440/560 (0) & -- & 410/560 (34) & 424/560 (21) & 440/560 (0)\\
\bottomrule
\end{tabular}
\end{table}

On the A100 the whole study again gave no false alarm on any post-fix version and the same \sys
detections, except \#1329 (54/80 in the original run) and \#1433 (86/200 in the rerun). The latter depends strongly on
precision, input family and $\beta$ (Appendix~\ref{app:evaluation-details}). Its cuBLAS GEMMs use the alignment model at the
\code{tl.dot} block, a hardware assumption; the A100's shorter block yields wider bounds. The \#1498 tests also detect all four deterministic inputs placed at half-ulp boundaries. Every post-fix call in these reruns passes (444 per GPU: 200 for \#1433, 44 for \#1498, and 200 for \#10616, including its 80 negative controls); the controls also pass before the fix.

The miss, \#986, adds a Hugging Face intermediate cast rather than repairing a demonstrated accuracy violation. Both versions pass our accuracy check. Uncertainty at the cast makes the current parity enclosure admit the difference (\cref{sec:quantizers}); PyTorch-referenced tolerances flag disagreement with that particular reference.

\para{vLLM} Before PR \#21783, \sys detected the offset bug in all 120 runs, against 75 for
$\mathit{atol} = \mathit{rtol} = 10^{-2}$ and 75 of 80 for the pre-fix test's own tolerance, which has no bf16 case.
The post-fix kernel still mis-decays the initial state when a sequence starts inside one chunk and continues into the
next, which upstream repaired later (PR \#23279); on such layouts \sys flagged the post-fix kernel in 83 of 90 runs and
the corrected one in none; $\mathit{atol} = \mathit{rtol} = 10^{-2}$ detected 51 paired cases and falsely
flagged 35 corrected runs.

\subsection{RQ5: Cost}\label{sec:rq5}

\para{Writing references} Excluding comments and blank lines, references take 2--12 lines each (generic) or
10--29 (Liger, including each kernel's precision plan). Mamba2's scan takes 14 lines, or 145 with helpers for
tiled dot products, TF32 conversion and state passing.

\para{Running the checks} On the RTX PRO 6000, Liger forward-and-backward calls take 87--168\,$\mu$s;
checking all outputs takes 4.9--13.8\,ms (median 83 times the kernel, range 29--127). Shadow-reference construction
has median 11\,ms; paired reference-plus-check time, 21\,ms across 21 configuration medians. The A100 is about
three times slower (kernels 252--478\,$\mu$s, checks 13.8--39.1\,ms), with similar ratios (median 80, range 29--116).
Excluding report generation, campaigns with all oracles took about 6.5--17 minutes per library on the RTX PRO 6000
and 24--43 on the A100.

\subsection{RQ6: Ablation}\label{sec:rq6}

RTX PRO 6000 ablations separately disable each refinement in Liger, Unsloth and vLLM, reusing kernel outputs
and rebuilding shadow references. The 200-trial CPU emulation reruns (\cref{sec:tightness}) cover deeper
references---GEMM, scaled FP8, MXFP4 and NVFP4 GEMMs, attention and cross-entropy---where intermediate errors weigh
more (\cref{tab:ablation}). Deliberately unsound controls treat broadcast copies as independent by ignoring keys,
or assume independence for components sharing events.

\begin{table}[!htbp]
\centering
\caption{Ablation. GPU: paired detections among changed-output runs; Liger/vLLM reruns; original Unsloth campaign; Liger kills out of 53 (Unsloth 25 and vLLM 41 for all sound settings). Correct kernels raise no false alarms; variants checked under full and unsound-independence settings also raise none. CPU: final-implementation median $B/B_{\mathrm{full}}$ over elements, then 200 seeds with two correct emulations each, with settings active during measurement. Quantized GEMMs use accuracy mode; other references are within 1\%.}\label{tab:ablation}
\footnotesize
\setlength{\tabcolsep}{4pt}
\begin{tabular}{lccccccccc}
\toprule
& & \multicolumn{4}{c}{GPU (RTX PRO 6000)} & \multicolumn{4}{c}{CPU: $B/B_{\mathrm{full}}$}\\
\cmidrule(lr){3-6}\cmidrule(lr){7-10}
Configuration & Sound & Liger & Liger kills & Unsloth & vLLM & FP8 & MXFP4 & NVFP4 & CE\\
\midrule
full \sys & yes & 93.4\% & 50 & 91.6\% & 79.7\% & 1 & 1 & 1 & 1\\
bound check only & yes & 88.4\% & 50 & 89.7\% & 77.7\% & 1 & 1 & 1 & 1\\
fixed-input mode & yes & 89.7\% & 50 & 87.0\% & 79.6\% & 2.8 & 2.2 & 2.2 & 4.4\\
no quadrature across origins & yes & 93.4\% & 50 & 91.4\% & 79.7\% & 1.50 & 1.43 & 1.45 & 1.52\\
no ordered products & yes & 93.4\% & 50 & 91.6\% & 79.7\% & 1.04 & 1.13 & 1.11 & 1.00\\
independence despite shared origins & no & 93.4\% & 50 & 91.6\% & 79.7\% & 1.00 & 1.00 & 1.00 & 1.00\\
keys ignored & no & \multicolumn{4}{c}{--} & 1.00 & 1.00 & 1.00 & 1.00\\
\bottomrule
\end{tabular}
\end{table}

Final-store half-ulps dominate our Triton bounds: correct-run worst elements use median bound fractions of 94\%
(Liger/Unsloth) and 78\% (vLLM). Random-error refinements change few verdicts: removing quadrature loses 16 of
6,177 changed-output Unsloth runs and one vLLM run. Two GPU components matter. Removing the rounding-interval
check loses 5.0, 1.9 and 1.9 percentage points of mutant-run detections on Liger, Unsloth and vLLM, and about
16\% of truncating-store detections on Liger (1,670 versus 1,997) and Unsloth (401 versus 480). Dropping zero mean
(fixed-input mode) loses 3.7 and 4.6 points on Liger and Unsloth and one Liger degenerate-row kill. Deeper
computations benefit more: quantized-GEMM and cross-entropy bounds widen 2.2--2.8 and 4.4 times without zero mean,
or 1.43--1.50 and 1.52 times without quadrature across disjoint origins, respectively. Removing ordered products
widens FP4 GEMM bounds 11--13\%.

Keys are necessary for soundness: ignoring them leaves tabulated median ratios at 1.00, but the correct
broadcast-then-reduce output (\cref{sec:linear}) exceeds its bound 5.6 times; full \sys uses 47\%. Independence
despite shared origins caused no false alarm: these references rarely sum such components; $XX^{\top}$ would.

\para{Choice of $\delta$} Changing $\delta = 10^{-6}$ to $10^{-3}$ flags only 5, 11 and 5 additional Liger,
Unsloth and vLLM runs with unchanged kills, since $\lambda$ grows with $\sqrt{\ln(1/\delta)}$. At $10^{-9}$, all
main correct-input cases remain supported except 120 of 480 vLLM cases exceeding the intermediate budget
(detection changes: Appendix~\ref{app:evaluation-details}). At $10^{-12}$, all 3,600 are unsupported. Increasing $\lambda_I$
restores budget admissibility but widens intermediate bounds.

\subsection{Threats to Validity}\label{sec:threats}

\emph{Seeded faults} may not represent real faults~\cite{just2014mutants}; nine historical cases supplement them.
\emph{Reference errors}: a development plan modeled a bf16 product as fp32; audit correction removed its false
alarms. \emph{Generality}: three libraries, two GPU generations, one compiler version; new platforms require
calibration, which tests but cannot prove A2. \emph{Statistics}: shared seeds and kernels limit independence;
Unsloth/vLLM mutant-level samples are small. \emph{A1}: zero false alarms do not prove the model. Affected-subset
GPU reruns gave 634 inconclusive outcomes per GPU (3.7\% of the original shifted-call schedule). Fixed-input mode
drops zero mean at a sensitivity cost. \emph{Tuned baselines} cover granularities, grids, safety factors and
percentiles, not every rule: row-statistic-dependent tolerances are another error model.
\section{Discussion}\label{sec:discussion}

\para{Detection scope} Per-element checks miss outputs inside acceptance regions (\cref{sec:single-run}).
Fault-exposing inputs and finer observations help: near-degenerate rows expose all three $\epsilon$ misplacements;
per-row cross-entropy exposes two faults hidden by averaging. Under the numerical assumptions, input-adaptive
bounds avoid tuned tolerances' observed false alarms here.

\para{Accuracy versus parity} The specification selects the mode. Parity preserves semantic casts with uncertain
inputs (Liger \#986); bit-exact replay additionally fixes operation order and rounding.

\para{Writing references} Overly strict precision plans risk false alarms; permissive plans risk missed faults.
Witnesses help investigate mismatches (\cref{sec:threats}); checking plans against kernel IR is future work.

\para{Kernels written by language models} KernelBench tests generated kernels on five random
inputs~\cite{ouyang2025kernelbench} with $\mathit{atol}=\mathit{rtol}=10^{-2}$ in fp16/bf16~\cite{kernelbenchrepo}.
Given a precision plan, \sys offers a signal with a model-conditional false-alarm rate.
\section{Related Work}\label{sec:related}

\para{Test oracles and approximate assertions} Poor numerical bounds~\cite{barr2015oracle} cause flaky
tests~\cite{dutta2020flaky}. FLEX uses extreme-value theory~\cite{dutta2021flex}; operator-aware calibration uses
error percentiles and safety factors~\cite{sarkar2026calibration}. Both learn from correct executions and may fail
to transfer (RQ3); \sys derives bounds from computation and profiles.

\para{Checking GPU kernels} Volta proves real-valued equivalence, abstracting floating-point
conversions~\cite{driscoll2026volta}. V-ABFT~\cite{gao2026vabft} extends A-ABFT's probabilistic matrix-checksum
thresholds~\cite{braun2014aabft} to mixed precision using variance estimates. These hardware-fault checks bound
single-operation checksums; \sys analyzes fused-kernel numerical error. KernelBench mutations expose tolerated
precision faults and cancellation-induced false alarms~\cite{du2026checker}. Shape/dtype
fuzzing~\cite{sarkar2026illusion} and formal verification~\cite{martinez2026frontier} complement this.
Cross-vendor tests expose numerical compiler differences, including math calls~\cite{zahid2024gpu}; we sample
selected profiles in calibration.

\para{Testing libraries} Backend differential testing~\cite{pham2019cradle}, API
fuzzing~\cite{wei2022freefuzz,deng2023titanfuzz}, compiler test generation~\cite{liu2023nnsmith} and numerical input
search~\cite{zhang2021predoo,yan2021grist} mostly use fixed-tolerance oracles that \sys could replace. Other studies
characterize numerical bugs and instability~\cite{difranco2017numerical,kloberdanz2022deepstability}.
Empirical studies of real-world numerical bugs in GPU programs~\cite{rathnasuriya2025gpu} and of code-generation
bugs in tile-based frameworks such as Triton~\cite{rathnasuriya2026tile} find numerical faults to be common. The
former identifies the lack of reliable test oracles as a key obstacle to detecting them, the gap \sys addresses for
low-precision kernels.

\para{Roundoff analysis} Static analyses bound worst-case scalar roundoff~\cite{solovyev2018fptaylor,darulova2018daisy},
also at scale~\cite{das2020satire}. Higher-precision shadows localize accuracy
problems~\cite{benz2012dynamic,sanchezstern2018herbgrind}; GPU instrumentation detects
exceptions~\cite{laguna2019fpchecker,li2023gpufpx}. Stochastic arithmetic perturbs rounding to estimate significant
digits~\cite{jezequel2008cadna,denis2016verificarlo}. Probabilistic analyses give normwise and componentwise
high-probability linear-algebra bounds~\cite{higham2019new,higham2020sharper,ipsen2020inner,connolly2021stochastic}
or analyze random-input programs~\cite{lohar2019sound,constantinides2021rigorous}. Componentwise analysis guides
inference accumulation precision~\cite{elarar2025mixed}; CESTAC-inspired analysis estimates operator
stability~\cite{chen2026stability}. \sys leaves reduction order open in derived tensor-program test oracles.

ProbTaylor applies concentration inequalities to scalar programs with independent input distributions~\cite{tao2026probtaylor};
its probability is over inputs, whereas ours fixes inputs and models rounding. Kernel Contracts specifies precision,
ordering and exceptional-value obligations~\cite{veit2026contracts}; \citet{shah2026verifier} build a contract-grade
verifier with adversarial gates and reduction-scale thresholds. These motivate derived tolerances; \sys composes
per-element regions, shared errors and a call-level probability budget, complementing determinism, exception and
shape checks.

\para{Low-precision arithmetic and hardware} Low-precision formats~\cite{micikevicius2022fp8,rouhani2023mx,nvidia2025nvfp4}
support LLMs; tensor cores may truncate~\cite{fasi2021tensor}; some FP8 GEMMs need periodic partial-sum
promotion~\cite{deepseek2024v3}. The aiter RFC derives tolerances from formats and reduction lengths~\cite{aiter5031};
\sys derives per-element bounds and calibrates its assumptions.
\section{Conclusion}\label{sec:conclusion}

\sys derives input-specific acceptance regions from reference programs and precision plans. Its calculus tracks
shared rounding errors, uses provenance to justify tighter product bounds and models low-precision quantizers
and hardware. The bound and rounding-interval checks turn these compositional rules into a model-conditional
false-alarm guarantee for an entire kernel call, while leaving permitted reduction orders open.

Main campaigns on two GPU generations raised no false alarm. On the RTX PRO 6000, \sys killed 116/124 effective
mutants, versus 59--83 for fixed baselines; historical replays exposed eight accuracy faults, with a ninth case
concerning rounding parity. Tuned-tolerance comparisons support input-specific bounds when input families
and row statistics change.

The final-version affected-subset reruns returned 634 inconclusive outcomes per GPU. These identify where
conservative envelopes prevent certification. Tighter local nonlinear bounds and automated checks of precision
plans could extend coverage; fault-exposing inputs and finer output observations remain complementary ways to
improve detection.

\section*{Data-Availability Statement}
An artifact provides source, references, mutants, notebooks and per-run records with campaign source
hashes. Scripts reproduce aggregates; GPU reruns require calibration. Appendix~\ref{app:evaluation-details} states data
exclusions.

\bibliographystyle{ACM-Reference-Format}
\bibliography{refs}

\appendix
\section{Proofs}\label{app:proofs}

\para{Setting} We work in the model of A1. Three kinds of operation are distinguished. An \emph{ordinary rounding}
of one element (an arithmetic rounding, one addition inside an accumulation, the division or grid step of a quantizer)
is a rounding \emph{event} $e$: it adds $\xi_e$ to the value $z_e$ it rounds, with $|\xi_e|$ at most half the spacing
of its format at $z_e$ and $\mathbb{E}[\xi_e \mid \mathcal{F}_e] = 0$, where $\mathcal{F}_e$ is the $\sigma$-algebra
generated by the errors of all earlier events. Events are totally ordered by the program order of $\Rp$ and, within
one operation, by a fixed order consistent with its data dependencies; for an accumulation, any topological
order of the implementation's summation tree. \emph{Deterministic roundings} (semantic casts and the computation of
quantizer scales) are $\fl_f$ of the model's values. \emph{Terminal stores}, the final rounding of each output of the
call, are $\fl_f$ of the model's pre-store value $z$; the stored value is what both checks observe, and it is not
reused inside the call. A rounding that an implementation skips, such as a product contracted into a fused
multiply-add, is an event with $\xi_e = 0$. A real round-to-nearest execution is the sample path $\xi_e = \fl_f(z_e) -
z_e$: the theorem is a statement about the model, and fixed-input mode makes the statement for every path. Half a
spacing at $z$ is at most $u_f|z|$ in the normal range and $\eta_f$ below it, so each event has the constant bound
$q_e = u_f(|v_e| + E_{\mathrm{hp},e}) + \eta_f$ for an arithmetic rounding whenever $|z_e| \le |v_e| + E_{\mathrm{hp},e}$,
$q_{\mathit{acc}}$ for an addition inside an accumulation, and the bounds of the quantizer case for its steps. The
shadow quantities $v$, $D$, $P$, $W$ and $\Ehp$ are functions of the input $x$ alone, never of the errors.
This also assumes that masks and indices of selection operations depend only on the inputs; that products inside
contractions are exact or fused, as for fp16, bf16, FP8 and TF32 operands; that a quantizer divides with a
correctly rounded fp32 division (a kernel that multiplies by a reciprocal needs a different plan); and, for the
triton\_fast profile, that operands of $\sqrt{\cdot}$, rsqrt, $\log$ and division are normal, since calibration
samples only normal inputs. A4 can fail where float64 evaluation cancels near the boundary of A3's domain, as in
$|y| - \Ehp(y)$ for a divisor, or where a format spacing jumps at a binade edge in semantic casts, quantizer grids
and scales; the 100-digit re-evaluation did not sample these cases.

\para{Checkpoints} Propagation checkpoints are the scalar analysis stages to which we assign a separate
martingale tail bound: internal ordinary-rounding results and materialized arithmetic, reduction, contraction,
elementary-function and quantization results. A terminal store's post-rounding value is not a propagation
checkpoint; its pre-store value is covered by propagation, and the stored output is handled only by
\cref{lem:terminal}. Deterministic semantic casts and value-copying shape or selection operations need no additional
tail event, although their summaries can be tighter than their operand's. Such an element is exact or a \emph{copy}
of one element $y$ of an earlier checkpoint: its error is $y$'s error plus deterministic cast residuals, which the
rules add to $D$, and its $\Wt$ equals $y$'s. Its components are $y$'s, relabeled and possibly merged; \code{where}
and \code{cat} may also enlarge origin sets or add components that are zero at this element. Those additions only
coarsen the overlap groups and cannot lower $R$, but merging two components with disjoint origins inside one group
can: unit proxies with origins $\{a\}$, $\{b\}$ and $\{a,b\}$ give $R=3$, and $R=\sqrt2+1$ once the first two merge. A merge only regroups the coefficients of the same martingale
sum $Z_y$, so the $R$ of every copy is also a valid proxy root for $Z_y$ (the merge rules applied to its groups, as for
the outputs below). We therefore give each checkpoint element $y$ the envelope $E^{*}_y = D_y + \min(W_y, \lambda_I
R^{*}_y) \le \Ehp(y)$, where $R^{*}_y$ is the smallest $R$ among $y$ and its copies; $R^{*}_y$ is a function of $x$
alone. If $|\varepsilon_y| \le E^{*}_y$, then every copy $z$ of $y$ satisfies $|\varepsilon_z| \le E^{*}_y + (D_z -
D_y) \le \Ehp(z)$. The local accumulation and quantization arguments below bound their internal partial sums, scale
computations, quotient and grid inputs from the operand checkpoints, without separate tail events for those internal
values.

Let $N_I$ be an upper bound on the number of propagation checkpoints, rather than the number of rounding events.
The implementation's \code{m\_int} is such a conservative count: ordinary rounding adds the number of result
elements; elementwise arithmetic and elementary functions add that number before any separate rounding;
quantization adds the number of reconstructed elements before output rounding; and cumulative sums add the
number of prefix outputs. A reduction of length $n$ adds $n$ times the number of output elements when it models
accumulation, and otherwise adds the number of outputs; a contraction of length $K$ adds $K$ times the number of
outputs. These latter counts also charge internal stages for which the local magnitude proof needs no separate
tail event. Shape operations, selections and semantic casts retain or sum their operand counts; the maximum rule
may charge additional elements although its deterministic enclosure needs no new tail. Counts are summed across
operands and across observed outputs, so reused ancestors can be counted repeatedly. Counts charged for terminal
stores may also be retained in $N_I$, even though these stores are excluded from the propagation checkpoints.
Both kinds of overcount only weaken the budget condition.

Let $G$ be the event that every propagation checkpoint $y$ is within its envelope $E^{*}_y$; on $G$, every
checkpoint and every copy is within its $\Ehp$. For an ordinary event $e$, let $G_e$ be the event that the
propagation checkpoints completed before $e$ obey their envelopes; the local operation rules then imply the operand
and internal-value envelopes used to bound event $e$, including those of operands that are copies. For a nonterminal element $y$, let
$H_y$ be the analogous event before $y$, without constraining $y$ itself. Then $G_e\in\mathcal F_e$,
$G\subseteq G_e$, and $H_y\subseteq G_e$ for every event $e$ affecting $y$.

\para{Truncation} Each event has an increment $\zeta_e$: $\zeta_e = \xi_e$, except where a case below splits off a part
whose indicator is $\mathcal{F}_e$-measurable (the subnormal part of an accumulation step). On $G_e$ the case bounds give
$|\zeta_e| \le q_e$. Define $\tilde\xi_e = \zeta_e\mathbf{1}[G_e]$. All indicators are $\mathcal{F}_e$-measurable, so
$\mathbb{E}[\tilde\xi_e \mid \mathcal{F}_e] = 0$, $|\tilde\xi_e| \le q_e$ everywhere, and $\tilde\xi_e = \zeta_e$ on $G_e$.

\para{Invariant} For every nonterminal shadow element $y$, on $H_y$ the implementation's error at $y$ equals
$d+\sum_c\tilde r_c$, with the properties below. Post-terminal summaries are only numerical envelopes for the
output checks; they do not assert a centered-sum representation of the deterministic store residual.
\begin{enumerate}
\item $|d| \le D$ ($d$ may be random; only its magnitude is bounded);
\item $\tilde r_c = \sum_{e \in \orig_c}\gamma_{c,e}\tilde\xi_e$, defined on the whole probability space, with signed
  coefficients $\gamma_{c,e}$ that are $\mathcal{F}_e$-measurable and bounded by constants, $|\gamma_{c,e}| \le
  \bar\gamma_{c,e}$;
\item $\sum_e \bar\gamma_{c,e}^2 q_e^2 \le P_c$ and $\sum_e \bar\gamma_{c,e}q_e \le W_c$;
\item for each event $e$, project the positions where its coefficient can be nonzero onto the coordinates outside
  $S_c$. At most $I_c$ distinct tuples occur. This deterministic support envelope is the component's
  \emph{global incidence cap}. In particular, a line of length $n$ along $d\notin S_c$ contains at most
  $\min(I_c,n)$ occurrences of $e$.
\end{enumerate}

\para{Preservation} We proceed by induction over the operations of $\Rp$. Each case uses only the bounds of its operands, which
hold on $H_y$, where $y$ is the element being built (the \emph{prefix event}).
\begin{itemize}
\item \emph{Rounding.} The new component has one event per element with coefficient 1, so $P = q^2$ and $W = q$.
  Events of distinct elements are distinct, so its global incidence cap is 1.
\item \emph{Scaling and broadcast.} Coefficients scale by $a$ and their bounds by $|a|$; this cannot enlarge support.
  A broadcast element's events now appear at every position along $d$, which is exactly what adding $d$ to the key
  records. Projecting away this shared axis leaves the old support, preserving its cap.
\item \emph{Merging in one element} (\cref{lem:merge}). If two components use disjoint event sets, the merged
  coefficient vector is their concatenation, so the proxies add. Otherwise the merged coefficient of each event is at
  most the sum of both, and Minkowski's inequality in $\ell^2$ gives
  \[
  \Big(\sum_e (\bar\alpha_e + \bar\beta_e)^2 q_e^2\Big)^{1/2} \le \Big(\sum_e \bar\alpha_e^2 q_e^2\Big)^{1/2}
  + \Big(\sum_e \bar\beta_e^2 q_e^2\Big)^{1/2}.
  \]
\item \emph{Incidence under coordinate operations.} Permuting axes remaps the key; a fixed roll or flip acts on one
  coordinate independently of the others. In both cases the projected support is merely relabeled. Basic slicing
  restricts support, relabels retained coordinates and projects away integer-indexed coordinates, so it cannot
  increase its cardinality. Singleton insertion is likewise harmless. These operations preserve $I_c$.

  For a reshape, each matched group of old and new axes has the same product of dimensions. If any old axis in the
  group is shared, all new axes of that group are marked shared; the group's coordinates disappear from the
  projected support. Otherwise the old and new group tuples are related bijectively. Thus each old unshared tuple
  determines at most one new unshared tuple, preserving $I_c$. Repeating an unshared axis $k$ times produces at most
  $k$ projected tuples per old tuple, so its cap is multiplied by $k$. Repeating a shared axis preserves the cap.

  Each component also carries a conservative token for its complete coordinate placement. Equal tokens certify a
  common support envelope for every shared event, independently of coefficient magnitudes and masks. Every
  coordinate-changing operation updates the token, including integer indexing on an axis other than the one later
  reduced. Ambiguous combinations receive fresh tokens. With disjoint event sets, a union has cap
  $\max(I_1,I_2)$ because any event comes from only one operand. For shared events at different placements, the
  union has cap at most $I_1+I_2$. Equal certified placements permit the maximum instead: masks and zero
  coefficients only restrict the common envelope. These rules cover component merges, elementwise operations and
  \code{where}, including merges caused by two remapped keys becoming equal. Concatenation puts its operands in
  different coordinate ranges. If any origins overlap, it therefore uses the sum of all operand caps even when
  incoming tokens agree; otherwise it uses their maximum. This remains conservative when the concatenation axis
  is shared.

  A general gather may choose a different index map at each setting of the other coordinates. If its axis $d$ is
  unshared and has output length $n$, each old projected tuple can produce at most $n$ new gather coordinates,
  while all other unshared coordinates remain unchanged. Multiplying $I_c$ by $n$ is therefore safe. The
  implementation needs this factor only for a nonidentity gather when some input axis is shared and $d$ remains
  unshared. With no shared axes and distinct indices within every line, gather is injective on positions and
  preserves the cap. If any line repeats an index, $d$ is added to the key; if $d$ was already shared, it stays
  shared. In either case projection removes the only coordinate that could multiply, preserving the cap. An
  identity gather also preserves it. A structured \code{index\_select} uses one index vector everywhere: distinct
  indices relabel or restrict projected support, and repeated indices are covered by adding $d$ to the key.

  More generally, adding axes to a key can only reduce projected support. The implementation may also replace a
  cap by its minimum with $\max(1,\prod_{d\notin S_c}n_d)$, the cardinality bound supplied by the tensor shape.
  Neither operation invalidates the invariant. Legacy per-axis multiplicities and tags are retained as metadata;
  the certified reduction factor uses the global cap and complete placement token above.
\item \emph{Reduction over $d$.} If $d \in S$, Minkowski over the $n$ positions gives
  $\sqrt{P'} \le \sum_i |a_i|\sqrt{P_i}$. If $d \notin S$, each event appears at no more than
  $m_d=\min(I_c,n)$ positions. For each event, Cauchy--Schwarz (\cref{lem:mult}) gives
  \[
  \left(\sum_i |a_i|\bar\gamma_{i,e}\right)^2
  \le m_d\sum_i a_i^2\bar\gamma_{i,e}^2,
  \qquad P'\le m_d\sum_i a_i^2P_i.
  \]
  The output event support is projected onto the retained coordinates, so its inherited global cap remains valid;
  the key and placement token are updated accordingly. This projection step is essential: a separate incidence
  bound on each input line need not remain valid on the new lines after reduction.

  A cumulative sum applies the same bound to each prefix, using the conservative factor $\min(I_c,n)$ for the full
  axis length $n$. It marks the prefix axis shared, so projected output support is contained in the input support
  projected away from that axis. Its inherited cap is preserved. Each accumulation event belongs to only one
  line, so its new component has cap 1, with the prefix axis shared for cumulative sums.

  In a matrix product, an inherited component is reduced along the contraction axis of length $K$ with factor
  $\min(I_c,K)$ when that axis is unshared, or coherently when it is shared. The new column axis for a left-operand
  component, or row axis for a right-operand component, is marked shared, as are any new broadcast batch axes.
  Projecting away these axes and the contraction coordinate leaves a projection of the old unshared support,
  preserving its cap. Component unions from the two operands obey the merging rule above; new accumulation
  events are private to each output position and have cap 1.
\item \emph{Accumulation.} For any summation tree, the computed sum equals the sum of the computed terms plus the
  rounding error $\xi$ of each addition, all with coefficient 1. On $H_y$ the rule's $\hat S$ bounds the absolute sum
  of the computed terms, from the operand checkpoints ($|v_i| + E_{\mathrm{hp},i}$ per term, and for a product the
  product of its operands' bounds). For a subtree
  with $k$ additions and leaf absolute sum $S_T$, $|\hat s_T| \le (1+u)^k(S_T + k\eta)$: the leaf case is $k = 0$, and
  for $k = k_L + k_R + 1$, $|\hat s_T| \le (1+u)(|\hat s_L| + |\hat s_R|) + \eta$, with both induction hypotheses and
  $(1+u)^{k_L+1}, (1+u)^{k_R+1} \le (1+u)^k$. Hence every partial sum, and the exact sum $s$ of the operands of every
  addition, is at most $(1+u)^n(\hat S + n\eta) \le e^{nu}(\hat S + n\eta)$ in magnitude, deterministically on the prefix
  event; no checkpoint is needed for them, and no condition $nu < 1$. Whether $s$ lies in the normal range is
  determined by earlier events, so $\mathbf{1}[s \text{ normal}]$ is $\mathcal{F}_e$-measurable, and the increment
  $\zeta_e = \xi_e\mathbf{1}[s \text{ normal}]$ has conditional mean zero and $|\zeta_e| \le u|s| \le q_{\mathit{acc}}$;
  these increments form the component. The remainder $\xi_e\mathbf{1}[s \text{ subnormal}]$ is at most $\eta$, and the
  $n$ remainders go to $D$ without using their mean.
\item \emph{Truncating accumulators} (A2). In the alignment model~\cite{fasi2021tensor}, the accumulator starts at
  zero, each block of at most $b$ products is formed exactly, the products and the running sum are aligned to the
  block's largest term and each is truncated to fp32 precision, and the block sum is truncated once more. Truncation
  never increases a magnitude, so by induction every term, partial sum and block sum is at most $\hat S$, and in the
  normal range one truncation loses less than one fp32 ulp of a value at most $\hat S$, at most $2u\hat S$. There are
  $K$ product truncations and two per block, so these errors total less than $(K + 2\lceil K/b\rceil)\cdot 2u\hat S$.
  Below the normal range the model assumes a gradual-underflow grid of spacing $2\eta$: each product's projection to
  that grid loses less than $2\eta$, after which the running sum and the products lie on the grid and their subnormal
  sums are exact, which gives the term $2K\eta$. The whole charge goes to $D$ and also covers the simpler model in which
  each block is summed exactly and truncated once. \emph{Scope of the evidence.} The charge is sound for an
  accumulator that behaves as this model with block $b$; that premise is A2. For Triton's \code{tl.dot}, calibration
  measured the sign of the bias and fitted $b$ on each GPU, an A100 stress test with products below $2^{-24}$ of the
  largest exceeded the one-truncation-per-block charge but not this one in any of 240 configurations, and a 32-case
  boundary grid passed on both GPUs. These are finite tests, not a proof of the hardware's behavior, and the study of
  \citet{fasi2021tensor} is an empirical characterization, not a specification. Applying the same model and $b$ to
  cuBLAS GEMMs (historical bug \#1433) or to \code{torch.\_scaled\_mm} is a further assumption that we did not test
  separately. Hardware that flushes subnormals, keeps a coarser internal floor or starts from a nonzero accumulator is
  outside this derivation.
\item \emph{Ordered products} (\cref{lem:ordered}). Use the identity
  $\hat X\hat Y - XY = (\hat X - X)Y + \hat X(\hat Y - Y)$. The first term multiplies $X$'s error by the exact
  constant $Y$. In the second, $\hat X$ depends only on events that precede all of $Y$'s events, provided the test
  includes the events folded into $\hat X$'s deterministic part ($\dorig$), since those still determine $\hat X$'s
  value. For each of $Y$'s events $e$, the coefficient $\hat X\,\mathbf{1}[|\hat X - v_X| \le \Ehp(X)]$ is therefore
  $\mathcal{F}_e$-measurable, bounded by $|v_X| + \Ehp(X)$, and equal to $\hat X$ on $H_y$. Deterministic parts give
  $D_X|v_Y| + (|v_X| + \Ehp(X))\,D_Y$.
\item \emph{Other products and division.} With shared or interleaved origins, $XY$'s error is
  $(\hat X - X)Y + X(\hat Y - Y) + (\hat X - X)(\hat Y - Y)$. The first two terms have constant coefficients, and the
  third is bounded by $\Ehp(X)\,\Ehp(Y)$ on $H_y$. Division uses the exact identity
  \[
  \frac{\hat x}{\hat y} - \frac{x}{y} = \frac{e_x}{y} - \frac{x\,e_y}{y^2} + \frac{x\,e_y^2}{y^2\,\hat y}
  - \frac{e_x\,e_y}{y\,\hat y},
  \]
  whose last two terms are bounded on the prefix event by
  \[
  \frac{|x|\Ehp(y)^2}{y^2m} + \frac{\Ehp(x)\Ehp(y)}{|y|m},
  \]
  where
  $m = |y| - \Ehp(y) > 0$ by A3. When $m \le 0$ the calculus yields no finite bound and the element is uncertified;
  the actual divisor may still be safe, as for $y = a - a + 2^{-10}$ with $a$ a rounded input, whose component the
  merge rule keeps.
\item \emph{Elementary functions.} For $f \in C^2([v - \Ehp, v + \Ehp])$, $f(v + e) - f(v) - f'(v)e =
  e^2\int_0^1(1-t)f''(v + te)\,dt$, of magnitude at most $\Ehp^2\sup|f''|/2$ on $H_y$ (\cref{eq:taylor}); the random
  coefficients are multiplied by the signed $f'(v)$ and their bounds by $|f'(v)|$. The library evaluates $f$ at a
  computed input in $I = [v - \Ehp, v + \Ehp]$, so A2 with the supremum over $I$, plus the floor $\max(\eta,
  2^{-96})$, bounds its error deterministically. A3 requires $I$ inside the domain; for $\sqrt{\cdot}$ an interval that reaches $0$ yields no
  finite bound.
\item \emph{Quantizer scales.} Scale computations are deterministic RN operations applied to the model's
perturbed inputs. For a block, let $A_b=\max_i|v_i|$ and $E_b=\max_i E_{\mathrm{hp},i}$. Its perturbed maximum is
in $[a_b^-,a_b^+]$, where $a_b^-=(A_b-E_b)_+$ and $a_b^+=A_b+E_b$. Define $A_g,E_g,a_g^-,a_g^+$ analogously
over the whole tensor. These conservative endpoints are the ones used by the implementation: the maximum is
1-Lipschitz in the infinity norm, so they enclose every permitted perturbed maximum.

For FP8 with scale format $f_s$, the scale map is
\[
 s_{\mathrm{FP8}}(a)=\max\{2^{-126},\fl_{f_s}(\fl_{32}(a/\max_f))\}.
\]
For MXFP4, the implemented map is
\[
 s_{\mathrm{MX}}(a)=2^{\lfloor\log_2\max(a,10^{-30})\rfloor-2}.
\]
The floor applies to all sufficiently small maxima, including zero. Both maps are non-decreasing and positive;
applying the relevant map to $a_b^-$ and $a_b^+$ gives $s_{\mathrm{lo}}$ and $s_{\mathrm{hi}}$.

For NVFP4 define the implemented global-scale map and its endpoints by
\[
 g(a)=\max\{10^{-30},\fl_{32}(a/2688)\},\qquad
 g_- = g(a_g^-),\qquad g_+ = g(a_g^+).
\]
Thus $0<g_-\le\tilde g\le g_+$. Let
\[
 T(a,g)=\fl_{\mathrm{E4M3}}(\fl_{32}(a/(6g))),\qquad
 J(t)=\begin{cases}1&t=0,\\t&t\ne0.\end{cases}
\]
The effective scale of a perturbed block is
$\tilde s=\fl_{32}(J(T(\tilde a_b,\tilde g))\tilde g)$.
Put $t_-=T(a_b^-,g_+)$ and $t_+=T(a_b^+,g_-)$. By monotonicity of RN and positivity of $g$, the raw block scale
lies in $[t_-,t_+]$. If $t_->0$, zero replacement cannot occur, so the effective block-scale factor is at least
$t_-$. If $t_-=0$, every nonzero E4M3 factor is at least $2^{-9}$ and the replacement factor 1 is no smaller.
Hence, with $\ell=t_-$ when $t_->0$ and $\ell=2^{-9}$ otherwise,
\[
 s_{\mathrm{lo}}=\fl_{32}(\ell g_-),\qquad
 s_{\mathrm{hi}}=\max\{\fl_{32}(J(t_+)g_+),\fl_{32}(g_+)\}
\]
enclose $\tilde s$. For the upper bound, a nonzero raw factor is at most $t_+$; a replaced factor is 1, covered by
the second term. These are the implementation's endpoint formulas, including its positive global floor, zero-block
replacement, and final fp32 scale-product rounding. Their evaluation is subject to A3--A4.
\item \emph{Quantization.} Under A1, the ordinary fp32 division gives
$\tilde r=\tilde x/\tilde s+\epsilon_d$. Let $M=\max_f$, $r=(|v|+\Ehp)/s_{\mathrm{lo}}$,
$q_d=u_{32}r+\eta_{32}$ and $r_d=r+q_d$. Then $|\epsilon_d|\le q_d$ and $|\tilde r|\le r_d$ on the operand
prefix event. Clip $\tilde r$ to $c=\min(M,\max(-M,\tilde r))$ and model its ordinary grid rounding by
$q=c+\epsilon_g$. The grid increment has magnitude at most $h_f(\min(r_d,M))$ by A1. Both increments have
conditional mean zero, and $\tilde s$ is fixed before either event. Before its separate reconstruction rounding,
\[
 q\tilde s-\tilde x
 =\tilde s\epsilon_d+\tilde s\epsilon_g+\tilde s(c-\tilde r).
\]
Let $W=s_{\mathrm{hi}}(q_d+h_f(\min(r_d,M)))$. The first two terms form a martingale sum with proxy at most $W^2$.
Their events are private to one quantized element, so the new component has global incidence cap 1 even when its
predictable scale coefficient is shared across a block.
If $r_d\le M$, clipping is impossible.
Otherwise, writing $C=(|v|+\Ehp-s_{\mathrm{lo}}M)_+$, the remaining term satisfies
\[
 |\tilde s(c-\tilde r)|
 =\tilde s(|\tilde r|-M)_+
 \le (|\tilde x|-\tilde sM)_++\tilde s|\epsilon_d|
 \le C+s_{\mathrm{hi}}q_d.
\]
The implementation charges the conservative amount $C+2s_{\mathrm{hi}}q_d$ to $D$ whenever $r_d>M$, which
covers this residual. A quantizer contract that bypasses grid rounding on a clipped element is also covered by
setting that grid increment to zero; the clipping indicator is known before the grid event. The subsequent
reconstruction rounding is an ordinary event unless it is a terminal store, in which case \cref{lem:terminal}
applies.
\item \emph{Maximum, absolute value and shifted exponentials.} $|\max_i \hat x_i - \max_i x_i| \le \max_i \Ehp(x_i)$
  and $||\hat x| - |x|| \le \Ehp(x)$ on $H_y$; both go to $D$, with the operands' origins recorded in $\dorig$. Softmax
  and log-sum-exp subtract a maximum before exponentiating. When the operand is exact, the kernel's maximum is the
  exact one. Otherwise the kernel subtracts the maximum of its own perturbed values, and the reference uses this
  maximum rule for the shift; it does not use the shift invariance of the exact function, so the bound may be looser
  but not smaller.
\item \emph{Semantic casts.} With $\hat y$ the computed value and $v$ the reference value,
  $\fl_f(\hat y) - \fl_f(v) = (\hat y - v) + (\fl_f(\hat y) - \hat y) + (v - \fl_f(v))$. The first term keeps its components;
  on the prefix event, $|\fl_f(\hat y) - \hat y| \le h(|v| + \Ehp)$ and $|v - \fl_f(v)| \le h(|v|)$, both deterministic
  residuals, not centered increments. A cast whose input interval reaches beyond the largest finite value yields no
  finite bound, and the cast's output tensor, with every value computed from it, is uncertified; an ordinary rounding,
  an elementary-function result or an accumulation whose interval may overflow makes its element uncertified, as does an
  element whose exact value is not finite. For a mean the overflow test uses the unscaled partial sums, which a
  sum-then-scale implementation accumulates.
\end{itemize}

\para{Tail bound} Let $Z = \sum_{k=1}^{m}\gamma_k\tilde\xi_k$ over events in order, with $\gamma_k$
$\mathcal{F}_k$-measurable and $|\gamma_k\tilde\xi_k| \le c_k = \bar\gamma_k q_k$, and let $P = \sum_k c_k^2$ and
$Z_j = \sum_{k \le j}\gamma_k\tilde\xi_k$. Given $\mathcal{F}_k$, the increment $X = \gamma_k\tilde\xi_k$ has mean zero
and lies in $[-c, c]$ with $c = c_k$. By convexity, $e^{tX} \le \frac{c + X}{2c}e^{tc} + \frac{c - X}{2c}e^{-tc}$, so
$\mathbb{E}[e^{tX} \mid \mathcal{F}_k] \le \cosh(tc) \le e^{t^2c^2/2}$ for every real $t$, the last step because
$(2j)! \ge 2^j j!$ (for $c = 0$, $X = 0$). This is the conditional Hoeffding lemma used by
\citet{connolly2021stochastic}. Since $Z_{k-1}$ is $\mathcal{F}_k$-measurable,
\[
\mathbb{E}\big[e^{tZ_k}\big] = \mathbb{E}\Big[e^{tZ_{k-1}}\,\mathbb{E}\big[e^{t\gamma_k\tilde\xi_k} \mid \mathcal{F}_k\big]\Big]
\le e^{t^2c_k^2/2}\,\mathbb{E}\big[e^{tZ_{k-1}}\big],
\]
and by induction $\mathbb{E}[e^{tZ}] \le e^{t^2P/2}$. For $P > 0$ and $a > 0$, Markov's inequality gives
$\Pr[Z > a] \le e^{-ta}\,\mathbb{E}[e^{tZ}] \le \exp(-ta + t^2P/2)$, which $t = a/P$ makes $\exp(-a^2/(2P))$; the same
holds for $-Z$. With $a = \lambda\sqrt{P}$, $\Pr[|Z| > \lambda\sqrt{P}] \le 2\exp(-\lambda^2/2)$. If $P = 0$, then $Z = 0$
and the event is empty.

\begin{lemma}[Terminal store]\label{lem:terminal}
Let an output's pre-store value $z$ have summary $(v, D, R, W)$, where $W=\Wt$ is the pre-store worst-case total. On $G$, let $z - v = d + Z$ with $|d| \le D$,
$|Z| \le W$ and $Z$ a martingale sum of proxy at most $R^2$. Let the stored value be $\fl_f(z)$, let $q = u_f(|v| + \Ehp)
+ \eta_f$, and let the calculus's ordinary rounding produce $(D, R^{+}, W + q)$. With $\lambda = \sqrt{2\ln(8N/\delta)}$
and $t = \sqrt{\lambda^2 - 1}$, both checks pass on $G \cap \{|Z| \le tR\}$, and $\Pr[G \cap \{|Z| > tR\}] \le
e^{1/2}\delta/(4N)$.
\end{lemma}

\begin{proof}
First $(R^{+})^2 \ge R^2 + q^2$: the new event is in no old component. If it forms its own group, $R^2$ grows by
exactly $q^2$; if it merges into a component of proxy $P$ in a group whose other components have total root proxy $S$,
the group grows by $(\sqrt{P + q^2} + S)^2 - (\sqrt{P} + S)^2 = q^2 + 2S(\sqrt{P + q^2} - \sqrt{P}) \ge q^2$, and no
other group changes. On $G$, $|\fl_f(z) - z| \le q$. On $G \cap \{|Z| \le tR\}$, by Cauchy--Schwarz,
$|\fl_f(z) - v| \le D + tR + q \le D + \sqrt{t^2 + 1}\sqrt{R^2 + q^2} \le D + \lambda R^{+}$, and always
$|\fl_f(z) - v| \le D + W + q$, so the bound check passes. Since $t < \lambda$ and $|Z| \le W$, also $|z - v| \le D +
\min(W, \lambda R) = B^{-}$, so $\fl_f(z) \in [\fl_f(v - B^{-}), \fl_f(v + B^{-})]$ by monotonicity and the interval
check passes. Finally, $\Pr[|Z| > tR] \le 2e^{-t^2/2} = e^{1/2}\cdot 2e^{-\lambda^2/2} = e^{1/2}\delta/(4N)$, and the
event is empty when $R = 0$. ($\lambda > 1$ because $8N/\delta > e^{1/2}$.)
\end{proof}

\begin{proof}[Proof of \cref{thm:soundness}]
For $\neg G$, consider the first checkpoint $y$ that leaves its envelope: the event is contained in $H_y \cap
\{|\varepsilon_y| > E^{*}_y\}$, and on $H_y$ the invariant gives $\varepsilon_y = d + Z_y$ with $Z_y$ defined
everywhere and $|Z_y| \le W_y$, so it is contained in $\{|Z_y| > \lambda_I R^{*}_y\}$. Because $R^{*}_y$ is the $R$
of a valid summary of $Z_y$, the tail bound applies with $\sqrt{P} \le R^{*}_y$; copies need no tail events of their
own. A union over the $N_I$ checkpoints gives
$\Pr[\neg G] \le N_I\cdot 2e^{-\lambda_I^2/2} \le \delta/2$, the premise of the theorem; this uses an unconditional tail
bound for a sum defined everywhere and never conditions a centered variable on $G$. Applying the merge rules to the
final groups shows that each output's pre-store error is $d + Z$ on $G$ with $Z$'s proxy at most $R^2$. For an output
with a terminal store, \cref{lem:terminal} contains the failure of either check, together with $G$, in an event of
probability $e^{1/2}\delta/(4N)$. For an output whose last step is a semantic cast, the cast keeps the components of
its operand and its two residuals are in $D$, so on $G$ both the bound check and the interval check (on the operand
$z$, with $|z - v^{-}| \le D^{-} + \min(\Wt, \lambda R)$) fail only on $\{|Z| > \lambda R\}$; an output with no final
rounding has only the bound check, which fails with $G$ only on the same event. Its probability is at most
$\delta/(4N)$. A union over the $N$ output elements of the call gives
$\delta/2 + e^{1/2}\delta/4 \approx 0.912\,\delta < \delta$. In fixed-input mode no probability is involved: every
$|\tilde r_c| \le W_c$ holds deterministically, so $G$ holds with $\Ehp = D + \Wt$, and with $B = D + \Wt$ neither
check can fail. Finally, A4 transfers the real-arithmetic bound to the float64 implementation through $\Bimpl$.
\end{proof}

\begin{proof}[Proof of \cref{lem:interval}]
Assume finite endpoints, no overflow and a fixed tie rule. $\fl_f$ is non-decreasing, so it maps $[a, b]$ into
$[\fl_f(a), \fl_f(b)]$, and the endpoints are attained by $a$ and $b$. For a representable $y$ with $\fl_f(a) < y <
\fl_f(b)$: if $y < a$, monotonicity and $\fl_f(y) = y$ give $y = \fl_f(y) \le \fl_f(a)$, a contradiction; similarly
$y > b$ would give $y \ge \fl_f(b)$. So $y \in [a, b]$ and $\fl_f(y) = y$ is attained.
\end{proof}

\begin{proof}[Proof of \cref{prop:single-run}]
If $K'(x) = K(x)$ for some $K \in \Cls$, then $O$ receives the same input and output in both cases. Since $O$ accepts
every element of $L(x)$, it passes $K'$ on $x$.
\end{proof}

\begin{proof}[Proof of \cref{prop:truncation}]
Let $f$ be one of the formats of \cref{sec:formats} and take $v^{-} > 0$; for negative values mirror the argument. Let $s$ be the spacing of $f$ at $v^{-}$, $\ell$ the
representable number with $\ell \le v^{-} < \ell + s$, $\varphi = (v^{-} - \ell)/s$, and $p \le 24$ the precision of
$f$, so $v^{-} < 2^{p}s$. Let $\varepsilon_p = 2^{-16} + 2^{p-38} < 10^{-4}$ and
$A' = \{5/8 + \varepsilon_p < \varphi < 7/8\}$, of probability $1/4 - \varepsilon_p$; on $A'$, $v^{-} > s/2$.
The implementation evaluates $v^{-}$ and $B^{-}$ in float64 as $\hat v$ and $\hat B$, and we use A4 in the form that
its allowance bounds these errors: $|\hat v - v^{-}| + |\hat B - B^{-}| \le a = 2^{-15}B^{-} + 2^{-40}v^{-}$. The
check's lower end $L$ evaluates $X = \hat v - (1 + 2^{-15})\hat B - 2^{-40}|\hat v|$ with three float64 roundings and
then steps once toward $-\infty$. Since $X = v^{-} - B^{-} - a + (\hat v - v^{-}) - (1+2^{-15})(\hat B - B^{-}) -
2^{-40}(|\hat v| - v^{-})$, and on $A'$ the roundings and the step move $X$ by at most $2^{-51}v^{-}$ in total,
\[
v^{-} - B^{-} - (2 + 2^{-14})\,a - 2^{-51}v^{-} \;\le\; L \;\le\; v^{-} + a + 2^{-51}v^{-}.
\]
With $B^{-} \le s/8$ and $v^{-} < 2^{p}s$, the loss on the left is at most
$\big((2^{-17} + 2^{p-39})(1 + 2^{-15}) + 2^{p-51}\big)s \le \varepsilon_p s$. On $A'$, every value within $B^{-}$
of $v^{-}$ lies in $(\ell + s/2, \ell + s)$ and is truncated to $\ell$, while $\ell + s/2 < L < \ell + 3s/2$.
Round-to-nearest maps $L$ to $\ell + s$, without a tie, also when $\ell + s$ starts a new binade, so the check's
lower end exceeds the truncated output and the element is rejected. With $H$ as in the proposition,
$\Pr[\text{reject}] \ge \Pr[A' \cap H] \ge 1/4 - 10^{-4} - \delta$. With independent positions, each of $N$
elements escapes with probability at most $3/4 + 10^{-4}$ outside the joint failure event of \cref{thm:soundness}.
\end{proof}

\section{Additional Evaluation Details}\label{app:evaluation-details}

\para{Witness development} We developed the calculus test-first: for each rule we searched for witness inputs on which an emulated correct
implementation exceeds the rule's bound, and each witness became a regression test. Five witnesses changed the rules
during development: inputs at rounding midpoints, where $\ulp^2/12$ underestimates the variance; legal intervals
spanning several ulps, which led to the rounding-interval rule; broadcast-then-reduce patterns, which led to keys;
the negative tail of PyTorch's GELU, which led to the $|z|$ envelope; and worst-case second-order product terms, which
led to the ordered-product rule. A review of the proofs found five
more: an accumulation so long that $1/(1-nu)$ turned negative, underflow in a two-term accumulation, a semantic cast
spanning many representable values, a product of a value with its own row maximum whose provenance had been lost in
$D$, and a library call whose internal $\exp$ overflow returned zero. They led to the rules for accumulation,
$\dorig$, semantic casts and the library floor. For those five repairs, none had caused a false alarm on GPU, and
no median bound grew over 10\%. A later audit corrected shared-event duplication through \code{where} and coordinate
shifts, quantizer half-spacing after division, nonfinite and overflow handling, softmax shifts of perturbed inputs,
and incidence counts under composed shape operations. Each repair has regression witnesses; \cref{sec:setup}
describes the associated reference comparisons and GPU reruns.

\para{Experiment versions} The mutation, tuning and timing results use the campaign-specific implementations
recorded in the artifact. LayerNorm, SSD, bugs \#1433/\#1498/\#10616 and the Liger/vLLM ablations were rerun with
the final implementation (e447e916); Unsloth ablations retain their original version. The final implementation adds
global incidence bookkeeping, the A3 finite-range check and the later witness repairs. For a separate CPU comparison,
we used snapshots 92d6e44 and e447e916, covering 730 reference calls across the four main shapes and sampled
families/configurations. Bounds changed only for LayerNorm, SSD and the three historical bugs in this comparison.
The largest sampled increase was 63.46\% in SSD's pre-store $\Bimpl$ (61.98\% at the terminal output); this is
not a universal bound on repair size. Equality on these samples does not establish identical GPU verdicts on all
other inputs. Both snapshots, detailed coverage and per-tensor comparisons are preserved in the artifact.

\para{Final guard patch} The artifact's numerical source is the final implementation plus a guard-only patch
(de63b01). The patch rejects unsupported fast-library output formats. It reports possible overflow inside elementary
functions and accumulations, including the unscaled partial sums of a mean, as uncertified; no finite bound changes.
No GPU run was repeated with the patch, and every run above keeps its recorded version. We rebuilt the reference
constructions of the main, stress, ablation, historical, boundary, shift and FP8 experiments with and without the
patch, using the recorded input schedules and precision plans: historical inputs regenerated on each GPU and matched to their
archived hashes, and SSD with the recorded TF32 truncation. All 41,758 records had bitwise-identical values, bounds,
rounding intervals and certification flags, and the T7 and T10 CPU campaigns produced byte-identical data.
RoPE uses only unchanged arithmetic paths, so GPU-specific trigonometric input generation does not affect this patch comparison.

\para{Distribution-shift certification} CPU reconstruction with the final implementation found 536 correct-input references without complete A3 certificates among 14,400 stress cases: 266 under S2 and 270 under S3. The corresponding variant sampling
schedule contributes 98 more, for 634/17,016 calls (3.7\%) per GPU. These are all rescaled fp16 gated activations:
320 GEGLU and 216 SwiGLU references; the other 13,864 correct references remain certifiable. S2 and S3 each contain
320 rescaled cases. The first range failures occur at terminal output rounding, not at a nonterminal fp16 cast.
For GEGLU, the global second-derivative bound $|\tanh''|\le0.77$ gives a loose envelope at large arguments; local
derivative bounds may reduce these inconclusive outcomes. An incomplete certificate prevents a certified pass;
the checker can still reject a different, certified element. We therefore reran all 634 affected calls with the
final implementation on each GPU. Both GPUs recorded 634 inconclusive outcomes and no failures: 536 correct
calls and 98 equivalent-variant calls. These are new GPU measurements on the recorded input schedule, not a
rescoring of the original GPU outputs or a rerun of the full 17,016-call campaign. The runner checks all 536
input hashes and the frozen source manifest; the returned plans, per-call records, summaries and logs agree.
The lightweight result archives contain per-call records and file hashes; input/output tensors and per-element
arrays are excluded from these archives.

This reconstruction evaluates references on CPU, not new GPU kernel executions. It uses the recorded seeds and
scales; 54 prespecified scale-search replays reproduce the recorded factors bit for bit. This does not prove
bitwise identity of every CPU/GPU input operation: RoPE trigonometric inputs, for example, were generated on GPU.
The separate SSD check uses each GPU's calibrated accumulator configuration. Per-case coverage and hashes are in
the artifact.

The additional shifted-input fault experiment uses 144 inputs and 216 changed-output pairs across five faults.
Re-evaluating saved GPU tensors retains 179 detections (17/36 under S4), compared with 49 and four faults for
$\mathit{atol}=\mathit{rtol}=10^{-2}$. Eleven correct inputs are inconclusive, but their mutants have unchanged outputs;
all 216 changed-output pairs have passing correct controls and are scorable.

\para{Additional mutation results} Truncating stores (M4) separate the oracles most clearly: \sys kills all 11 Liger M4 mutants and the best fixed oracle
kills 3; the rounding-interval check flags more of their runs (\cref{sec:rq6}), but the bound check alone also finds each. Eight Liger mutants are killed by \sys alone: six
truncating stores, an extra rounding of the cross-entropy log-sum-exp, and a truncated store of the per-row loss. The
last two are visible only with reduction set to \code{none}; with per-row losses \sys flags 138 and 157 of 240 runs
and every fixed oracle none. In Unsloth \sys alone kills an extra rounding inside RMSNorm, and in vLLM four mutants
(an extra rounding, three low-precision scales).

Per run, \sys flags 85\% of bf16 and 83\% of fp16 runs of effective Liger mutants, against at most 73\% for any fixed
oracle; in fp32 it flags 43\%, like \code{assert\_close} and the aiter formula, because the triton\_fast profile's
library allowances dominate the tiny fp32 rounding errors.

\para{Missed mutants} \sys misses eight mutants, for three reasons. \emph{A loose fp32 bound.} The three $\epsilon$ misplacements change bf16 and fp16 outputs by less
than the final rounding. In fp32 their worst element lies a median 5--13 times further from the float64 reference than the
correct kernel's, yet uses at most 52\% of the bound, which the math-library allowance dominates there;
\code{assert\_close} kills them on 31--38 fp32 inputs, and \sys detects every degenerate-row instance
(\cref{sec:rq3}). The extra rounding of the Liger loss also stays within the bound (at most 68\% of it, against 39\%
for the correct kernel); only tuned tolerances kill it. \emph{No loss of accuracy.} The Unsloth RoPE
change is closer to the float64 reference than the correct kernel (50\% of the bound against 95\%), and no fixed oracle kills
it. \emph{Not distinguished by our checks.} The three vLLM changes leave the worst-element error exactly as large as the correct
kernel's. Fixed oracles kill two of them only on 1--4 of 240 inputs, crossing a threshold at which they flag 46--95\%
of correct runs. Two are truncating stores whose bias toward zero shows in 95--97\% of runs, suggesting a pooled test under a separate model of correct-kernel bias (\cref{sec:single-run}).

\para{Percentile calibration} Operator-aware calibration sets each tolerance to 1.5 times the 95th percentile of the
maximum absolute error over correct runs~\cite{sarkar2026calibration}. We computed it per operator, dtype and output
from the stored per-output statistics. Calibrated on the tuning phase, as before a release, it flags 32\% of correct
Liger runs, 31\% of Unsloth and 41\% of vLLM runs against the PyTorch reference, and 90--100\% on degenerate rows;
calibrating per operator and dtype only, or against the float64 reference, gives 23--54\%.
Calibrated on the evaluation inputs themselves, it still flags 2.5\%, 2.2\% and 4.6\% of correct runs (2.0--6.7\% over these choices) and 88--100\% on
degenerate rows, and kills 47 of 53, 24 of 27 and 28 of 44 mutants, against \sys's 50, 25 and 41.

\para{Exact ablation populations} Table~\ref{tab:ablation} counts paired detections: the correct control passes and the mutant fails on the same input. Its main-input changed-output denominators are 11,405 Liger runs, 6,177 Unsloth runs and 8,579 vLLM runs. Unchanged outputs of effective mutants are excluded from these percentages; including them gives denominators 15,120, 7,680 and 10,560 instead. These are run counts, not independent mutant counts.

For M4, the changed-output denominators are 2,088 Liger and 480 Unsloth runs. Removing the interval check loses 327 of 1,997 and 79 of 480 detections, respectively: relative losses of 16.4\% and 16.5\%, rather than a 16-percentage-point change in all-run detection rate. Removing quadrature loses 16 Unsloth and one vLLM detection; removing ordered products or enabling the unsound independence control changes no paired detection in these GPU campaigns (Liger/vLLM reruns; original Unsloth campaign).

\begin{table}[!hbp]
\centering
\caption{Paired detections by ablation setting, using the changed-output populations stated above.}\label{tab:ablation-counts}
\begin{tabular}{lrrr}
\toprule
Configuration & Liger & Unsloth & vLLM\\
\midrule
Full & 10,651 & 5,660 & 6,836\\
Bound check only & 10,080 & 5,543 & 6,670\\
Fixed-input & 10,225 & 5,375 & 6,828\\
No quadrature & 10,651 & 5,644 & 6,835\\
No ordered products & 10,651 & 5,660 & 6,836\\
Unsound independence & 10,651 & 5,660 & 6,836\\
\bottomrule
\end{tabular}
\end{table}

At $\delta=10^{-3}$, paired detections increase by 5, 11 and 5, with no losses and the same mutant kills. At $10^{-9}$, Liger detections are unchanged; Unsloth loses three (5,660 to 5,657), retaining 25 kills. vLLM drops from 6,836 to 5,040: 1,794 earlier detections become unsupported and two become passes. Its kills drop from 41 to 40. In total, 2,319 of 8,579 changed-output vLLM runs are unsupported at that budget. Unsupported results are neither passes nor detections. The paper's separate 120/480 budget count refers to main correct-input vLLM cases, not these mutant runs.

\para{GPU reruns} Both GPUs ran the final implementation with 20 main seeds, 10
degenerate/probe seeds and four shapes. The 10,290 non-tuning control calls per GPU comprise 360 LayerNorm
correct calls, 3,600 LayerNorm variants, 480 SSD correct calls, 5,760 SSD variants and 90 residual-bug controls.
All their outputs passed both certified checks, with no inconclusive result; 90 additional tuning calls also passed.
LayerNorm retains one main-input miss (an epsilon misplacement), and its two non-compiling mutants account for 720 crash calls per GPU,
excluded from numerical kills. SSD retains three missed mutants. The original offset bug is detected on all 120
affected main inputs; the residual post-fix bug is detected on 83/90 probes, with 0/90 flags on the corrected control.
The residual-bug probes compare the still-faulty version after the first fix with its matched corrected control.

The historical reruns contain 534 before/after pairs per GPU: 444 evaluation and 90 calibration pairs. Every
post-fix output passes the certified checks. Statistical checks can be inconclusive when there are too few
independent units; these are separate from certified outcomes. For \#1498, 40/40 Gaussian cases and four deterministic inputs with the LayerNorm mean at a half-ulp boundary detect the bug in fixed-input mode. For \#10616, all 120 affected cases are detected and both versions pass the
80 negative controls. The configuration-specific \#1433 counts follow (each cell is detections out of 20).
Its cuBLAS alignment charge remains a hardware-model assumption, not a proof from \code{tl.dot} calibration.

\begin{table}[!htbp]
\centering
\caption{Detections for historical bug \#1433, out of 20 inputs per cell.}\label{tab:jsd-configurations}
\begin{tabular}{llrrrr}
\toprule
& & \multicolumn{2}{c}{RTX PRO 6000} & \multicolumn{2}{c}{A100}\\
dtype & $\beta$ & Gaussian & High logits & Gaussian & High logits\\
\midrule
bf16 & 0 & 20 & 20 & 12 & 20\\
bf16 & 0.5 & 14 & 20 & 0 & 20\\
bf16 & 1 & 20 & 20 & 5 & 20\\
fp16 & 0 & 0 & 14 & 0 & 9\\
fp16 & 1 & 0 & 0 & 0 & 0\\
\bottomrule
\end{tabular}
\end{table}

The Liger and SSD ablation reruns completed all 1,920 and 480 planned cases, respectively. Held-out
counts (819 Liger and 234 SSD correct inputs) are filters on main, not additional trials. All supported correct
and checked equivalent outputs pass. At $\delta=10^{-9}$, 78/234 held-out SSD controls are unsupported and the
held-out kill count is 36/44; at $10^{-12}$ neither campaign has a scored correct input. Recorded mutant-output
differences relative to the original campaign affect 20 Liger and 181 SSD comparisons. Every ablation setting
uses the same newly recorded output and paired correct control; changed-output populations are derived from
that campaign, so the new SSD denominator is 8,579 rather than the original 8,640. These differences are not
attributed solely to the bound repair. The distributed archive contains case metadata and verdict summaries; per-element arrays from these ablation runs are not included.

\end{document}